\documentclass[a4paper,12pt,english]{article}
\usepackage[utf8]{inputenc}   
\usepackage{lmodern}
\usepackage[english]{babel}   
\usepackage[babel]{microtype}
\usepackage{amsfonts,amsmath,latexsym,amsthm,fixmath}
\usepackage{amssymb,ifthen}
\usepackage{authblk}
\usepackage{enumitem}
\usepackage{color,graphicx}
\usepackage{url}
\usepackage[
    pdftex,
    pdfstartview={XYZ 0 1000 1.0},
    bookmarks=false,
    colorlinks,
    breaklinks=true,
    citecolor=citecolor,
    filecolor=blue,
    linkcolor=linkcolor,
    urlcolor=urlcolor,
    pdfborder={0 0 0}
]{hyperref}
\definecolor{urlcolor}{rgb}{0, 0.5, 0}
\definecolor{citecolor}{rgb}{.5,0,.25}
\definecolor{linkcolor}{rgb}{0,0,1}

\usepackage{geometry}
\graphicspath{{figures/}}
\def\svgpath{figures/}

\newtheorem{theorem}{Theorem}

\newtheorem{proposition}[theorem]{Proposition}

\newtheorem{lemma}[theorem]{Lemma}
\newtheorem{note}[theorem]{Note}

\theoremstyle{remark} 
\newtheorem{remark}{Remark}

\newcommand{\includesvg}[2][]{%
\def\tempa{#1}\def\tempb{}%
\ifx\tempa\tempb\else\let\svgwidth\tempa\fi
\input{\svgpath#2.pdf_tex}%
}

\graphicspath{{./figures/}}
\def\svgpath{figures/}

\title{A Universal Triangulation for Flat Tori} 

\author[1]{Francis Lazarus \thanks{This author is partially supported by the French ANR projects GATO (ANR-16-CE40-0009-01) and MINMAX (ANR-19-CE40-0014) and the LabEx PERSYVAL-Lab (ANR-11-LABX-0025-01) funded by the French program Investissement d’avenir.}}

\author[2]{Florent Tallerie} 

\affil[1]{G-SCOP, CNRS, UGA, Grenoble, France}
\affil[2]{G-SCOP, UGA, Grenoble, France}

\newcommand{\E}{\mathbb{E}}
\newcommand{\Z}{\mathbb{Z}}
\newcommand{\N}{\mathbb{N}}
\newcommand{\R}{\mathbb{R}}
\newcommand{\C}{\mathbb{C}}
\newcommand{\Hyp}{\mathbb{H}}
\newcommand{\Torus}{\mathbb{T}}
\newcommand{\Sphere}[1]{\mathbb{S}^#1}
\newcommand{\SL}{\text{SL}_2(\Z)}

\newcommand{\optV}{1217}
\newcommand{\optT}{2434}
\newcommand{\VlongT}{135}
\newcommand{\TlongT}{270}
\newcommand{\TshortT}{836}

\newcommand{\Lmin}{33}
\newcommand{\ie}{\textit{i.e.}}
\definecolor{definecolor}{rgb}{0,0.1,0.55}
\def\define#1{\textbf{\textcolor{definecolor}{#1}}}
\usepackage[colorinlistoftodos]{todonotes}

\makeatletter
\def\myref#1{%
  \expandafter\ifx\csname r@#1\endcsname\relax
    0\@latex@warning{Reference `#1' on page 
              \thepage \space undefined}%
  \else
    \ref{#1}%
  \fi}
\makeatother

\begin{document}

\maketitle

\begin{abstract}
A result due to Burago and Zalgaller (1960, 1995) states that every orientable polyhedral surface, one that is obtained by gluing Euclidean polygons, has an isometric piecewise linear (PL) embedding into Euclidean space $\E^3$. A flat torus, resulting from the identification of the opposite sides of a Euclidean parallelogram, is a simple example of polyhedral surface.
  In a first part, we adapt the proof of Burago and Zalgaller, which is partially non-constructive, to produce PL isometric embeddings of flat tori. Our implementation produces embeddings with a huge number of vertices, moreover distinct for every flat torus. In a second part, based on another construction of Zalgaller (2000) and on recent works by Arnoux et al. (2021), we exhibit a \emph{universal triangulation} with \optT{} triangles which can be embedded linearly on each triangle in order to realize the metric of any
  flat torus.
\end{abstract}

\section{Introduction}
\label{sec:Introduction}

A celebrated theorem of Nash~\cite{n-c1ii-54} completed by Kuiper~\cite{k-oc1ii-55} implies that every smooth Riemannian orientable surface has a $C^1$ isometric embedding in the Euclidean 3-space $\E^3$. As a consequence one can represent and visualize faithfully in $\E^3$ the geometry of any abstract orientable Riemannian surface. A $C^1$ isometric embedding of the square flat torus was for instance constructed by the HEVEA team~\cite{bjlt-fttds-12}. An analogous result, due to Burago and Zalgaller~\cite{bz-iplit-95}, states that every orientable polyhedral surface, obtained by abstractly gluing Euclidean polygons, has an isometric piecewise linear (PL) embedding in $\E^3$. In particular, this provides PL isometric embeddings for every flat torus, the result of the identification of the opposite sides of a Euclidean parallelogram. However, the proof of Burago and Zalgaller is partially non-constructive, relying on the subdivision of the polyhedral surface into an acute triangulation and on the Nash-Kuiper theorem itself, which is a priori far from constructive. The singular vertices of the polyhedral surface (where the angles at the incident polygons do not sum up to $2\pi$) moreover deserve special treatments with several constants that are rather hard to estimate.  In the case of flat tori, all these difficulties can be circumvented. In particular, a flat torus has no singular vertex. Using a simple construction of acute triangulations together with the conformal embeddings of  Hopf-Pinkall~\cite{p-hts3-85,b-ghmpt-88}, we were able to compute PL isometric embeddings of various flat tori, including the square and the hexagonal tori.

Our implementation of the construction of Burago and Zalgaller, even including our simplifications for flat tori, produces PL embeddings with a huge number of vertices: more than 170,000 for the square torus and more than 7 millions for the hexagonal torus. Most importantly, the underlying triangulations of the resulting PL embeddings depends on the geometry (or \emph{modulus}) of the flat tori and are pairwise non-isomorphic. Apart from the construction of Burago and Zalgaller, describing explicit PL embeddings of specific flat tori does not seem a simple task. As an illustrating example, it is only very recently that a triangulation composed of 80 triangles and allowing to embed isometrically any rectangular torus, including the square flat torus, appeared in the literature~\cite{q-eples-20}.

We say that a triangulation of the topological torus is \emph{universal} if, for any flat torus, it admits a geometric realization in $\E^3$ that is isometric to this flat torus. It is not clear that such a universal triangulation should exist as the moduli space of flat tori is not compact. In particular, there is no reason why any of the triangulations obtained from the method of Burago and Zalgaller would be universal. Our main result is the rather counterintuitive existence of a universal triangulation with the description of such a triangulation of reasonable size.

\begin{theorem}\label{th:univ-triang}
  There exists an abstract triangulation $\cal T$ of the torus with \optT{} triangles
  that admits for each flat torus (in the moduli space) an embedding in $\E^3$ which is linear on each triangle of $\cal T$,
and which is isometric to this flat torus. Moreover, every flat torus has an isometric PL embedding in $\E^3$ with at most 270 triangles.
\end{theorem}
Our proof combines a construction of Zalgaller~\cite{z-sblc-00} for the isometric embedding of so-called \define{long tori} with very recent works by Tsuboi~\cite{t-oeft-20} and Arnoux et al.~\cite{alm-iplef-21} for embedding flat tori. We use the latter construction for the isometric embedding of \define{short tori}. Define the \emph{aspect ratio} of a flat torus as the ratio of its area by the square of the length of its shortest closed geodesic. A flat torus is called long or short depending on whether its aspect ratio is respectively large or small. We construct explicit universal triangulations for long and short tori, see Propositions~\ref{prop:long-tori} and~\ref{prop:short-tori}.
Moreover, we show in Section~\ref{subsec:short-tori} (Lemma~\ref{lem:short-tori}) that every short torus has a geometric realization with 76 triangles. Together with Proposition~\ref{prop:long-tori} this implies that every flat torus has a geometric realization with at most 270 triangles. To prove the first part of the theorem we finally overlay the universal triangulations for long and short tori to obtain a universal triangulation for all flat tori.

\paragraph{Organization of the paper}
After introducing the necessary definitions in Section~\ref{sec:background}, we outline in Section~\ref{sec:BZ95} the general construction of Burago and Zalgaller~\cite{bz-iplit-95} for the isometric embedding of polyhedral surfaces. We apply this construction to flat tori in Section~\ref{sec:embedding-flat-tori} to produce nice pictures of their PL realizations. Section~\ref{sec:universal-triangulation} is devoted to the proof of Theorem~\ref{th:univ-triang}. We first construct a universal triangulation for long tori in Section~\ref{subsec:long-tori}, then recall the diplotori approach of Arnoux et al. in Section~\ref{subsec:diplotori} to produce a universal triangulation for short tori in Section~\ref{subsec:short-tori}. The formal proof of Lemma~\ref{lem:short-tori} claiming that we can indeed cover all the short tori with three diplotori is deferred to the Appendix.
We finally deduce a universal triangulation for flat tori in Section~\ref{subsec:merging-short-long}. 

\section{Background and definitions}\label{sec:background}
\paragraph{Polyhedral surfaces}
A \define{polyhedral surface} is a compact topological surface obtained from a finite collection of polygonal regions in the Euclidean plane by gluing their sides according to a partial oriented pairing. This pairing should be such that each side is paired at most once and two sides in a pair should have the same length. The pair orientation specifies one of the two isometries between its sides.
Note that two sides of a same polygon may well be glued together. The resulting surface is \emph{closed}, \ie, without boundary, when each side appears in one pair, \ie, when the pairing is complete. Since every polygon can be triangulated, one can replace the polygons by triangles in this definition. The collection of triangles together with their gluing determine a \define{triangulation} of the surface. This triangulation is \define{simplicial} when there is no loop edge or parallel edges, \ie, when no triangle side has its endpoints identified by the gluing and when any two triangle sides share a connected or empty subset after the gluing.
By an \define{abstract triangulation} of a surface, we mean a simplicial complex whose underlying space is homeomorphic to that surface. Note that an abstract triangulation that admits an embedding in $\E^3$ that is linear on each triangle and that is isometric to
a polyhedral surface must be isomorphic to some simplicial triangulation of the polyhedral surface. 

\paragraph{Polyhedral metric}
The gluing of Euclidean polygons induces a \define{length metric} on the resulting polyhedral surface: the distance between any two points is the infimum of the lengths of the paths connecting the two points. Here, we consider paths that are finite concatenations of subpaths contained in a single polygon and the length of a path is the sum of the Euclidean length of these subpaths. There is an intrinsic definition of polyhedral surfaces that does not
assume any specific decomposition into polygons. Formally, a
\define{polyhedral metric} on a topological surface is a metric such
that every point has a neighborhood isometric to a neighborhood of the
apex of a Euclidean cone, where we ask that the isometry sends the
considered point to the apex of the cone. In turn, a (2-dimensional)
Euclidean cone is defined by coning  from the origin a rectifiable simple (non
self-intersecting) curve lying on the unit sphere in $\E^3$.
The length of this curve is the total angle of the cone. A point whose
conic neighborhood has total angle different from $2\pi$ is called a
\define{singular vertex}. Note that in any triangulation of a
polyhedral surface by Euclidean triangles the singular vertices must
be vertices of the triangles.

\paragraph{Piecewise linear maps and isometries}
Let $S$ be a polyhedral surface. A map $f: S\to \E^3$ is called \define{piecewise linear} (PL) if $S$ admits a triangulation such that the restriction of $f$ to any triangle is \emph{linear}, \ie, it preserves barycentric coordinates. Once a triangulation of $S$ is given, the image of its vertices in $\E^3$ determines a unique \define{linear\footnote{As often the case in the literature, we use tha adjective ``linear'' to qualify a map that preserves a specific triangulation, and reserve the terminology ``piecewise linear'' for a map defined on a triangulable space without a prescribed triangulation, which is linear for \emph{some} triangulation.} map} on this triangulation by extending linearly to the images of triangles. 

$f$ is \define{piecewise distance preserving} if $S$ admits a triangulation such that the restriction of $f$ to any triangle is distance preserving, \ie,  $|f(x)-f(y)| = d_S(x,y)$ if $x$ and $y$ lie in the same triangle. Here, $|\cdot|$ is the Euclidean norm and $d_S$ is the polyhedral metric on $S$. In particular, the piecewise distance preserving map $f$ must be \define{length preserving}: if $\gamma: [a,b]\to S$ is a rectifiable path, then $\gamma$ and its image $f\circ\gamma$ have the same length. The map $f$ is an \define{embedding} if it induces a homeomorphism onto its image $f(S)$ endowed with the restriction of the topology of $\E^3$. In that case, $f(S)$ is naturally equipped with a length metric induced by the Euclidean metric of $\E^3$ so that the length of a path in $f(S)$ is its Euclidean length as a path in $\E^3$.

A length preserving embedding is the same as an \define{isometric embedding} between $S$ and $f(S)$, where each surface is endowed with its own length metric, respectively polyhedral and induced by the Euclidean metric. Thus, a piecewise distance preserving embedding is the same as a \define{PL isometric embedding}.
A map $f: S\to \E^3$ is called \define{contracting} if there is a constant $C<1$ such that $|f(x)-f(y)|\leq Cd_S(x,y)$ for all $x,y\in S$.
\begin{remark}\label{rk:short-embedding}
  If $S$ is orientable, then it admits a contracting embedding into $\E^3$; simply embed $S$ into $\E^3$ as your favorite surface model with the same genus as $S$ and compose with a scaling of sufficiently small ratio to contract all distances. The model can be polyhedral or of any desired regularity.
\end{remark}

\paragraph{Flat tori}
A \define{flat torus} is a polyhedral surface obtained from a Euclidean parallelogram by pairing its opposite sides. We usually consider flat tori up to re-scaling since multiplying all the distances by the same constant does not modify the essential geometric properties of the torus. This amounts to consider that similar parallelograms lead to the same flat torus. If  $(e_1,e_2)$ is the canonical  basis of the Euclidean plane, we can thus assume that the two sides of the parallelogram are respectively $e_1$ and $\tau$ for some vector $\tau=\tau_1 e_1+\tau_ie_2$, with $\tau_i>0$. Identifying the real plane with the complex line, we conclude that a flat torus is determined by its \define{modulus} $\tau = \tau_1+i\tau_i$.

Rather than gluing the sides of a parallelogram, one can equivalently obtain the same flat torus by quotienting the Euclidean plane $\E^2$ by the rank 2 lattice $\Z \tau+\Z e_1$ acting by translations. The same lattice is generated by the vectors $(a\tau + b, c\tau + d)$, where $\begin{pmatrix} a & b \\ c & d\end{pmatrix}\in \SL$ is an integer matrix with determinant 1. After applying the adequate similarity, this lattice corresponds to the modulus $(a\tau +b)/(c\tau + d)$, where $\tau$ is again viewed as a complex number. In fact, the set of flat tori is in one-to-one correspondence with the quotient $\Hyp^2/\SL$, where $\Hyp^2$ denotes the upper half-plane (the set of moduli) and $\SL$ acts as above. Every flat torus has a modulus in the fundamental domain of this quotient as shown in Figure~\ref{fig:moduli-space}. The reader may consult~\cite[Sec. 12.2]{fm-pmcg-12} for a more in-depth introduction to the space of flat tori, the so-called modular surface.
\begin{figure}[h]
  \centering
 \includegraphics[width=.5\textwidth]{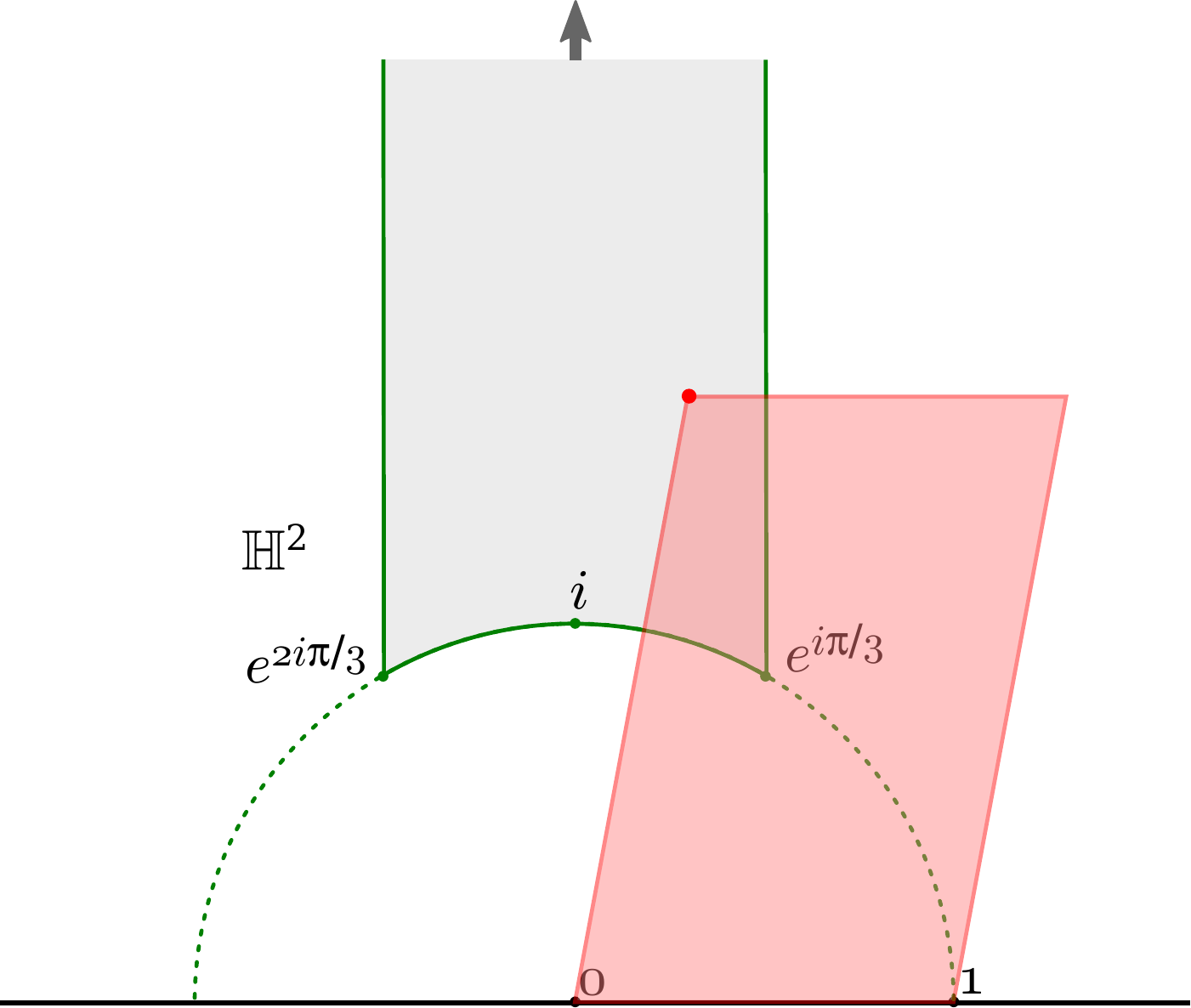}
  \caption{A point (in red) in a fundamental domain of the moduli space of tori (in light grey) with the corresponding parallelogram. Note that the boundary of the domain should be identified adequately to represent this moduli space faithfully, leading to a non-compact orbifold.}
  \label{fig:moduli-space}
\end{figure}

\section{The construction of Burago and Zalgaller}\label{sec:BZ95}
We first recall the result of Burago and Zalgaller for embedded surfaces.
\begin{theorem}[Burago and Zalgaller \cite{bz-iplit-95}]\label{th:BZ95}
  Every contracting $C^2$ embedding in $\E^3$ of a polyhedral surface can be approximated by a PL isometric embedding. 
\end{theorem}
Here, the approximation by a PL isometric map means that for any $\varepsilon>0$ there is such a map moving the points of the contracting $C^2$-embedding\footnote{As a topological surface, a polyhedral surface admits a unique smooth structure compatible with the Euclidean structure at the non-singular points. We can thus speak of a $C^2$ embedding of the surface.} by a distance less than $\varepsilon$. By Remark~\ref{rk:short-embedding} this implies that every orientable polyhedral surface has an isometric PL embedding in 3-space.
Before we give a sketch of the proof, we describe the basic construction of Burago and Zalgaller, which is a specialization of Theorem~\ref{th:BZ95} to the case of a single triangle.

\subsection{Embedding a triangle}\label{subsec:base-case}
We recall that a triangle is \define{acute} if its three internal angles are less than $\pi/2$. If $t$ is a triangle in $\E^3$ and $\vec{n}$ is a vector normal to $t$, then the \define{prism above} $t$ is the set $\{p+\lambda \vec{n}\mid p\in t, \lambda \geq 0\}$ and the three infinite faces of this prism are its \define{walls}.
\begin{lemma}[\cite{bz-iplit-95}]\label{lem:BZ95}
  Let $T=A_1A_2A_3$ and  $t=a_1a_2a_3$ be (Euclidean)  triangles in $\E^3$ such that
  \begin{enumerate}[label=\bfseries(\roman*)]
    \item\label{it:triang-1} $T$ and $t$ are acute,
  \item\label{it:triang-2} $|a_ia_j| < |A_iA_j|$ for $i,j=1,2,3; i\neq j$,
    \item\label{it:triang-3} the distance of the circumcenter $\omega$ of $t$ to each side $a_ia_j$ is smaller than the distance of the circumcenter $\Omega$ of $T$ to the corresponding side $A_iA_j$.
    \end{enumerate}
Denote by $m_{ij}$ the point in the wall above $a_ia_j$ at equal distance $|A_iA_j|/2$ from $a_i$ and $a_j$.
Then, $T$ has a PL isometric embedding in the prism above $t$ (with respect to one of the two possible normal directions) with the boundary condition that each side $A_iA_j$ is sent to the broken line $a_im_{ij}a_j$.
\end{lemma}
This Lemma (see Figure~\ref{fig:prism})
\begin{figure}[h]
  \centering
  \includegraphics[width=.2\linewidth]{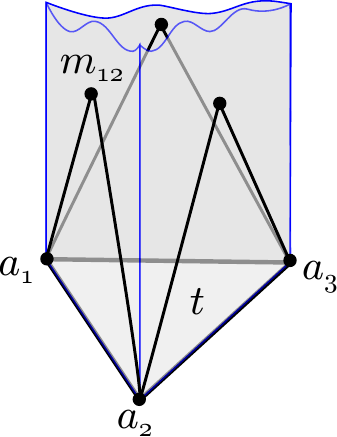}
  \caption{The prism above $t$.}
  \label{fig:prism}
\end{figure}
easily implies that $T$ has a PL isometric embedding arbitrarily close to $t$. Indeed, by subdividing $T$ and $t$ uniformly as in  Figure~\ref{fig:uniform-subdivision} we get similar triangles of smaller size to which we can individually apply Lemma~\ref{lem:BZ95}. Thanks to the boundary condition in the lemma, the individual constructions fit together to form an isometric embedding of $T$. The constructions for the smaller triangles being homothetic to the construction for the original triangles, we get closer and closer to $t$ as we refine the uniform subdivisions.
\begin{figure}[h]
  \centering
  \includegraphics[width=.6\linewidth]{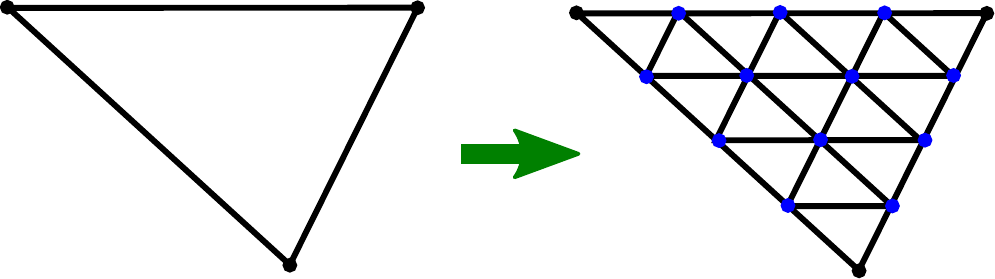}
  \caption{Uniform subdivision of a triangle. The vertices of the subdivision have barycentric coordinates $(i/n,j/n,k/n)$ for $i,j,k\in \N$ and $i+j+k=n$ for some fixed $n$.}
  \label{fig:uniform-subdivision}
\end{figure}

The triangles $T$ and $t$ being acute, they contain their circumcenters $\Omega$ and $\omega$ in their interior.  We let $\vec{n}$ be a unit vector normal to $t$ and we let $\omega'$ be the point vertically above $\omega$ such that $|a_1\omega'|=|A_1\Omega|$. Refer to Figure~\ref{fig:BZ2-fold} for an illustration.  Note that $\omega'$ is well-defined since by the assumptions \ref{it:triang-2} and \ref{it:triang-3} the circumradius $|A_1\Omega|$ of $T$ is larger than the circumradius $|a_1\omega|$ of $t$.  
For completeness, we recall the proof of Lemma~\ref{lem:BZ95}. Triangle $T$ is first subdivided into three subtriangles $\Omega A_i A_j$. The goal is to fold each $\Omega A_i A_j$ above $\omega a_ia_j$ with the boundary condition for  $A_iA_j$ as in the lemma and so that the boundary edges $\Omega A_i, \Omega A_j$ are sent respectively to the segments $\omega'a_i$ and $\omega'a_j$. To this end, we first fold $\Omega A_1A_2$ along its altitude from $\Omega$ and place the resulting two-winged shape above $t$ so that the side $A_1A_2$ is folded onto the broken line $a_1m_{12}a_2$. 
\begin{figure}[h]
  \centering
  \includegraphics[width=.6\textwidth]{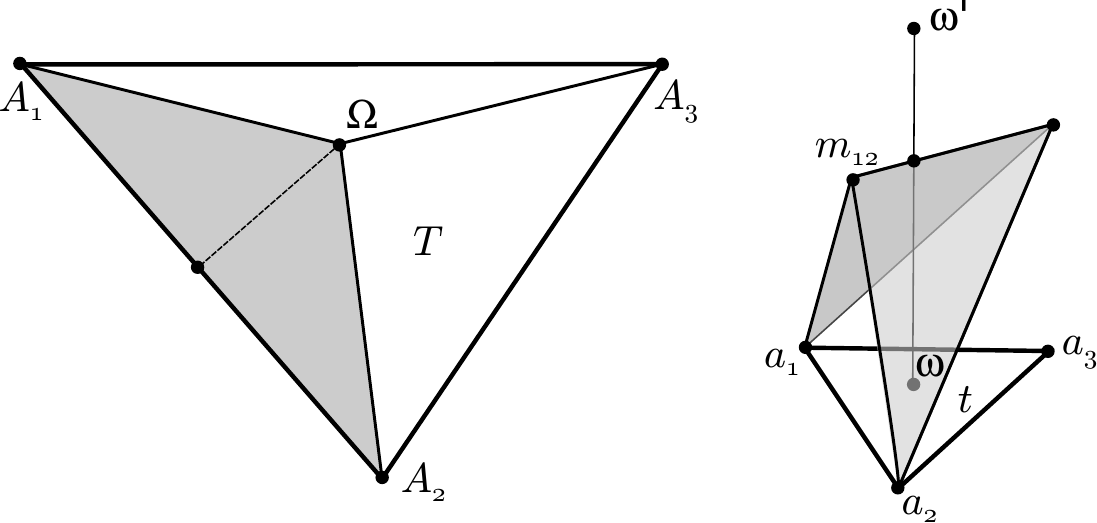}
  \caption{The subtriangle $\Omega A_1A_2$ is folded above $t$.}
  \label{fig:BZ2-fold}
\end{figure}
We next consider a plane $\Pi_1$ in the pencil generated by $a_1a_2$ to reflect the part of the two-winged shape lying to the right of that plane. See Figure~\ref{fig:folding-2}.
\begin{figure}[h]
  \centering
  \includegraphics[width=.85\textwidth]{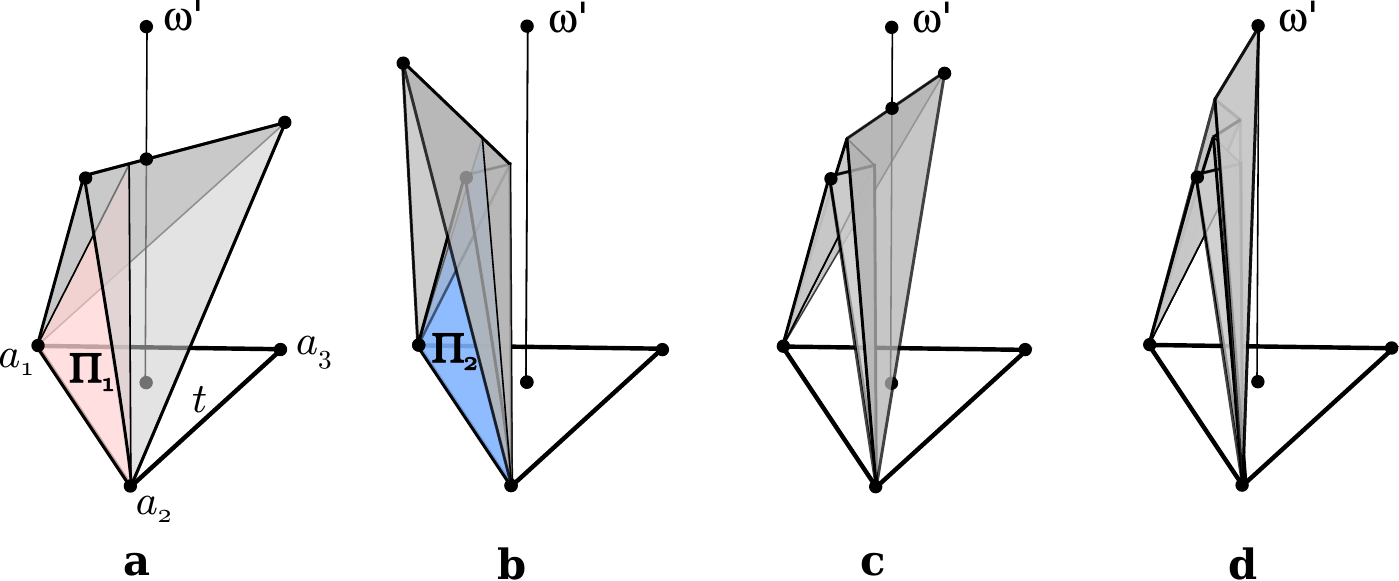}
  \caption{a, the reflection plane $\Pi_1$. b, after reflection in  $\Pi_1$, and the plane $\Pi_2$. c, reflection in $\Pi_2$. d, after an even number of reflections the point $\Omega$ is sent to $\omega'$. }
  \label{fig:folding-2}
\end{figure}
Another plane $\Pi_2$ in the same pencil is then chosen to reflect part of the already reflected part. Choosing $\Pi_1$ and $\Pi_2$ appropriately, it is not hard to see that after an even number of such reflections the point $\Omega$ in $\Omega A_1A_2$ will be sent to $\omega'$. We finally apply the same construction to  the two other subtriangles $\Omega A_2A_3$ and $\Omega A_3A_1$ and paste them to form a folding of $T$ above $t$ as desired.

\begin{note}\label{nt:folding}
This folding of $T$ admits some flexibility. In particular, the boundary conditions can be modified so that each boundary wedge $a_im_{ij}a_j$ is tilted around the axis $a_ia_j$. This modification is needed in order to paste the constructions above two adjacent triangles that are non coplanar as illustrated on Figure~\ref{fig:BZ2-adjacent}.
\end{note}
\begin{figure}[h]
  \centering
  \includegraphics[width=.7\textwidth]{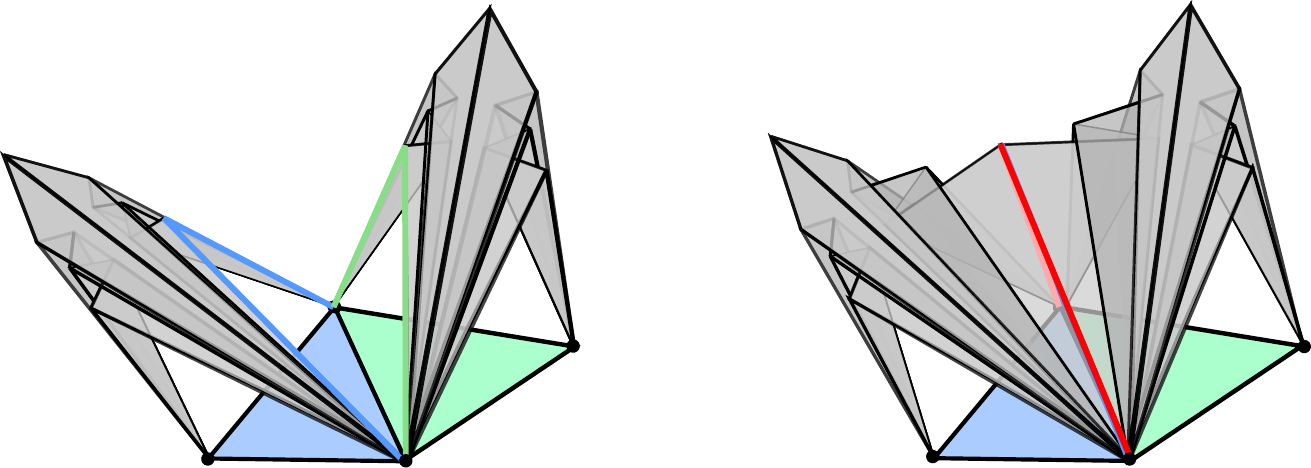}
  \caption{Pasting two foldings of large triangles sharing an edge above smaller triangles that are non coplanar.}
  \label{fig:BZ2-adjacent}
\end{figure}

\subsection{Embedding arbitrary polyhedral surfaces}
Denote by $f: S\to\E^3$ the contracting $C^2$ map in Theorem~\ref{th:BZ95}. Let $U$ be a union of small polygonal disks centered at each singular vertex of $S$. The strategy for the proof of Burago and Zalgaller is the following. 
\begin{enumerate}[label=\bfseries(\alph*)]
\item\label{it:strgy-1} Compute an acute triangulation of $S\setminus U$, where each triangle is acute.
\item\label{it:strgy-2} Compute an approximation $f_1$ of $f$ that is almost conformal on $S\setminus U$ and contracting over $S$. Here, \emph{almost} conformal means that $f$ almost preserves angles, or more formally that its coefficient of quasi-conformality, or dilatation~\cite[Section 11.1.2]{fm-pmcg-12}, is close to one.
\item\label{it:strgy-3} Refine the acute triangulation of $S\setminus U$ uniformly to obtain an acute triangulation $\cal T$ with small triangles. The meaning of \emph{small} depends on the geometric properties of $f_1$ and on the flexibility in Note~\ref{nt:folding}.
\item\label{it:strgy-4} Replace $f_1$ by its PL approximation $F$ mapping linearly each triangle $T = A_1A_2A_3$ of $\cal T$ to the triangle $F(T) := f_1(A_1)f_1(A_2)f_1(A_3)$ in $\E^3$.
\item\label{it:strgy-5} Apply the construction in Section~\ref{subsec:base-case} to every pair $(T,F(T))$, using the tilted version in Note~\ref{nt:folding} in order to paste the constructions of adjacent triangles.
\item\label{it:strgy-6} Fill the gaps corresponding to $U$ with specific constructions to deal with singularities as described in~\cite{bz-iplit-95}, depending on whether the conical angle of a singularity is smaller or larger than $2\pi$.
\end{enumerate}
We comment on the above steps, referring to the original paper~\cite{bz-iplit-95} for more details.
Computing an acute triangulation as required in Step~\ref{it:strgy-1} is a non-trivial task. If $S$ is obtained from a gluing of Euclidean triangles it was shown how to compute an acute refinement~\cite{bz-prd-60,z-stdat-13} of reasonable size. Step~\ref{it:strgy-2} is the most challenging and relies on the Nash-Kuiper theorem~\cite{n-c1ii-54,k-oc1ii-55}. The idea is to apply this theorem in order to approximate $f$ by an almost isometric map with respect to a metric homothetic to the metric on $S$, but slightly smaller. This provides the map $f_1$ that is at the same time contracting and almost conformal. This almost conformality combined with the uniform subdivision in Step~\ref{it:strgy-3} implies that any triangle $T$ of $\cal T$ is sent by the PL approximation $F$ of Step~\ref{it:strgy-4} to an almost similar triangle $F(T)$. Since $f_1$ and its PL approximation are contracting, this in turn implies that the pair $(T,F(T))$ satisfies Conditions \ref{it:triang-1},\ref{it:triang-2} and \ref{it:triang-3} in Lemma~\ref{lem:BZ95}. Moreover, the fact that the triangles in $\cal T$ are small together with the smoothness of $f_1$ ensure that $F$ maps adjacent triangles to almost coplanar triangles.
We can thus apply the basic construction of Lemma~\ref{lem:BZ95} and its tilted version as in Note~\ref{nt:folding} to perform Step~\ref{it:strgy-5}. This eventually lead to a PL isometric embedding of $S\setminus U$. It remains to embed appropriately the neighborhood of the singular vertices as required by Step~\ref{it:strgy-6} to complete the PL isometric embedding of $S$. 

\section{Embedding flat tori}\label{sec:embedding-flat-tori}
Step~\ref{it:strgy-2} in the proof of Burago and Zalgaller is highly non constructive, and to our knowledge no explicit PL isometric embedding of a closed surface according to their method was known up to date. It appears that the steps of their construction can be greatly simplified in the case of flat tori. Thanks to these simplifications we were able to visualize PL isometric embeddings of various flat tori in $\E^3$.

We first observe that there is no need for Step~\ref{it:strgy-6} since a flat torus has no singular vertex: the angles at the four corners of its defining parallelogram add up to $2\pi$, showing that the only vertex after the side gluing is non singular. In particular, one should set $U=\emptyset$ in all the steps.
\subsection{Acute triangulation of flat tori}\label{subsec:acute-triang}
Itoha and Yuan~\cite{iy-atft-09} have shown  that every flat torus can be triangulated into at most 16 acute triangles\footnote{We emphasize that we are not trying to show the existence of a universal acute triangulation for flat tori, i.e., a triangulation isomorphic to an acute triangulation on every flat torus.  This would easily follow from Itoha and Yuan~\cite{iy-atft-09}. However, there is no reason why such a triangulation would admit a linear isometric embedding for every flat torus. In fact, acute triangulations only appear as an intermediate step in the construction of Burago and Zalgaller, but their final geometric realizations are not acute.}. However, since we need a fine triangulation as in Step~\ref{it:strgy-3} with a good control on the acuteness, we use the following triangulation, which is conceptually simpler. Let $\tau$ be the modulus of the flat torus $\Torus_\tau:= \E^2/(\Z\tau + \Z e_1)$ (we abusively identify the plane with the complex numbers). We consider the equilateral triangular lattice generated by $e^{i\pi/3}/n$ and $1/n$ for some positive integer $n$. This lattice comes with a regular triangulation ${\cal T}_e$ by equilateral triangles. Let $p_{a,b}=ae^{i\pi/3}/n + b/n$, with $a,b\in \Z$, be a point in this lattice that is closest to $\tau$. In particular, $|\tau - p_{a,b}|\leq (n\sqrt{3})^{-1}$. We deform ${\cal T}_e$ by a linear transformation $\ell$ defined by  $1\mapsto 1$ and $p_{a,b}\mapsto \tau$.  By the previous inequality and for $n$ large enough, $\ell$ is close to the identity. The triangles in $\ell({\cal T}_e)$ are thus close to equilateral. Now, the lattice $\Z\tau + \Z e_1$ leaves $\ell({\cal T}_e)$ invariant, so that $\ell({\cal T}_e)/(\Z\tau + \Z e_1)$ is a well defined triangulation of $\Torus_\tau$ by almost equilateral triangles. See Figure~\ref{fig:equi-triang}.
\begin{figure}[h]
  \centering
  \includegraphics[width=.9\textwidth]{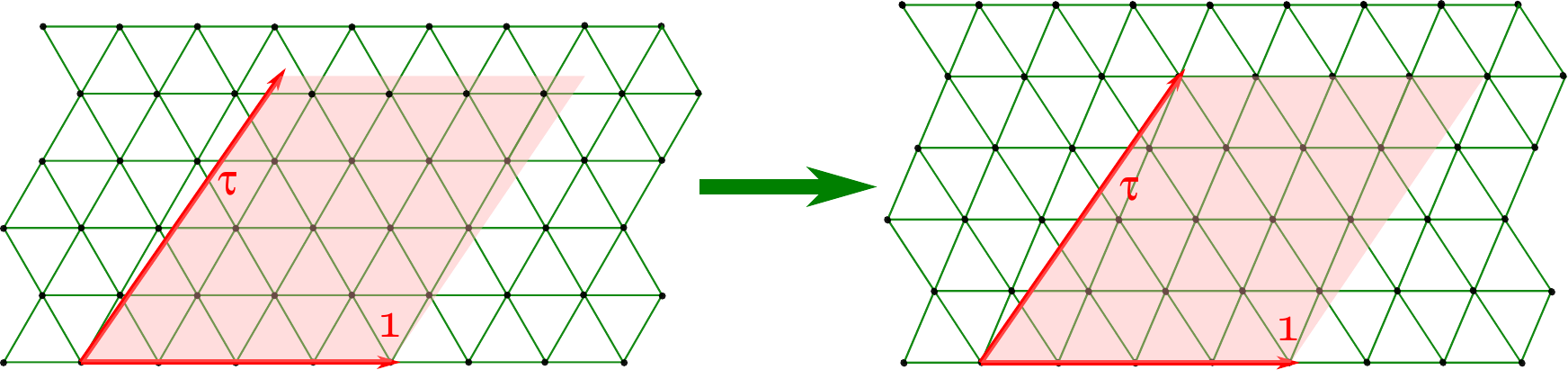}
  \caption{The equilateral triangular lattice (here with $n=4$) is deformed to fit the lattice of $\Torus_\tau$.}
  \label{fig:equi-triang}
\end{figure}

\subsection{Conformal embedding of flat tori}\label{subsec:conformal}
Theorem~\ref{th:BZ95} requires an initial contracting $C^2$ embedding, further approximated in Step~\ref{it:strgy-2} by an almost conformal map. In the case of flat tori we can directly provide a contracting conformal embedding. 
\paragraph{The case of rectangular tori}
For conciseness we identify $\R^3$ with $\C\times\R$. First observe that the standard embedding of the square torus as a torus of revolution, $f: \Torus_1\to\E^3, (u,v)\mapsto \big((R+r\cos(2\pi u)) e^{i2\pi v}, r\sin(2\pi u)\big)$, is never conformal as the ratio of the lengths of the partial derivatives is non-constant.
The partial derivatives are nonetheless orthogonal and when the torus is rectangular, \ie, when $\tau=i\tau_i$ is pure imaginary, there are conformal maps of the form $f_{i\tau_i}: \Torus_{i\tau_i}\to\E^3, (u,v)\mapsto f(\alpha(u),v/\tau_i)$ for some $1$-periodic function $\alpha$.
Indeed, when $\alpha$ satisfies the differential equation $\tau_i\alpha'=\cos(2\pi \alpha) + \frac{R}{r}$, one easily checks that the partial derivatives of $f_{i\tau_i}$ with respect to $u$ and $v$ have the same norm (and are orthogonal). This differential equation solves to
\[ \alpha(u) = \frac{1}{\pi} \arctan\left( \sqrt{\frac{R+r}{R-r}} \tan \left( \frac{\sqrt{R^2-r^2}}{\tau_i r} \pi u \right) \right)\quad
\text{with}\quad \frac{R}{r} = \sqrt{\tau_i^2 k^2 + 1}
\]
for some integer $k$. In practice, we have chosen $k=r=1$, leading to the conformal map:
\[ f_{i\tau_i}(u,v) = \left(\left(\sqrt{\tau_i^2+1}+\cos\left(2\pi \alpha\left(u\right)\right)\right) e^{i2\pi v/\tau_i}, \sin\left(2\pi \alpha\left(u\right)\right)\right)
  \]
  with $\alpha(u) = \frac{1}{\pi} \arctan\left( \sqrt{\frac{\sqrt{\tau_i^2 + 1}+1}{\sqrt{\tau_i^2 + 1}-1}} \tan \left( \pi u \right) \right)$. It remains to compose $f_{i\tau_i}$ with a contracting homothety to get a contracting conformal embedding of $\Torus_{i\tau_i}$. See Figure~\ref{fig:conformal-rect-torus} for an example.
  \begin{figure}[h]
    \centering
    \includegraphics[width=.9\linewidth]{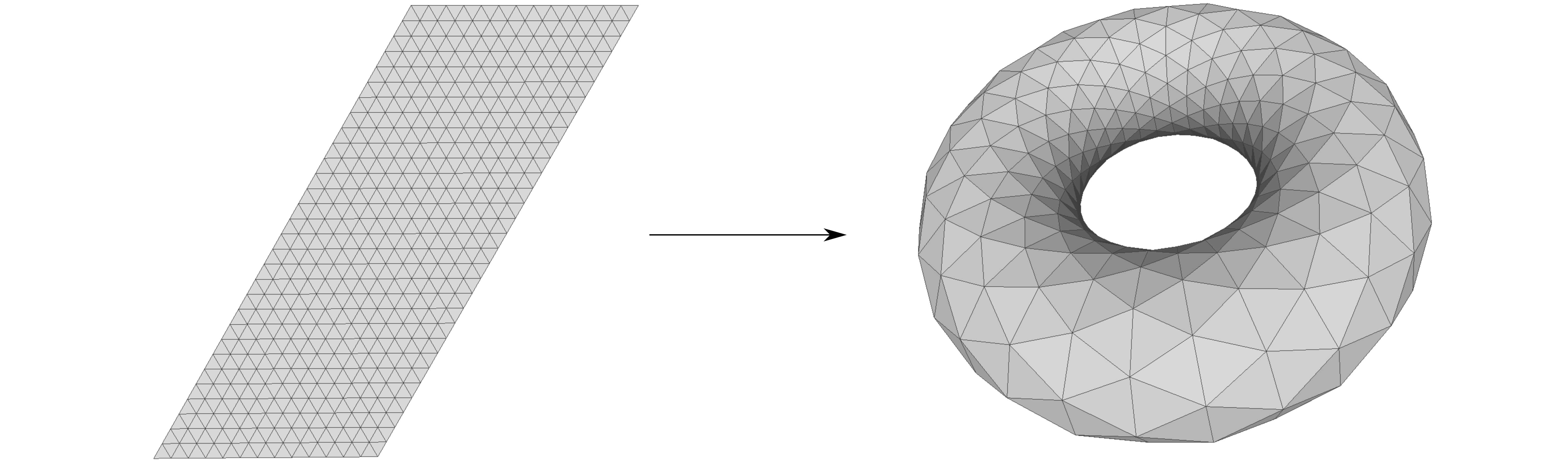}
    \caption{Left, the fundamental domain of the rectangular torus $\Torus_{2i}$ is chosen to be aligned with an almost equilateral tiling as in Section~\ref{subsec:acute-triang}. Beware that the identification of the horizontal sides includes a non-integral translation. Right, PL approximation of the conformal embedding $f_{2i}$. }
    \label{fig:conformal-rect-torus}
  \end{figure}
  \paragraph{The general case}
  In order to embed non rectangular tori conformally we rely on the Hopf tori developed by Pinkall~\cite{p-hts3-85}. These are based on the Hopf fibration
  \[p: \Sphere{3}\to\Sphere{2}, (x,y,z,t)\mapsto (2xz+2yt, 2xt-2yz, x^2+y^2-z^2-t^2),
  \]
  a standard projection of the 3-sphere $\Sphere{3}$ onto the 2-sphere $\Sphere{2}$ whose fibers (the sets $p^{-1}(s)$ for $s\in \Sphere{2}$) are circles. Pinkall proves that if $\gamma$ is a simple closed curve on $\Sphere{2}$, then $p^{-1}(\gamma)$ is a flat torus isometric to $\Torus_\tau$ with $\tau = (A+iL)/(4\pi)$, where $L$ is the length of $\gamma$ and $A$ is the oriented area delimited by $\gamma$ on $\Sphere{2}$, choosing the side of $\gamma$ so that $A\in [-2\pi,2\pi)$. Since this torus lies in $\Sphere{3}\subset \E^4$, it remains to apply a stereographic projection, say from the South pole $(0,0,0,-1)$, assuming it does not lie on the torus, to obtain a conformal embedding of $\Torus_\tau$ in $\E^3$. In coordinates: $(x,y,z,t)\mapsto (x,y,z)/(t+1)$.
  
  Banchoff~\cite{b-ghmpt-88} revisited Pinkall's approach to give explicit parametrizations of the Hopf-Pinkall tori. On $\Sphere{2}$, Banchoff considers a curve of the form $\gamma_\tau(\theta) = (\sin \phi(\theta) e^{i\theta}, \cos \phi(\theta))$ given in spherical coordinates, where the polar angle $0 <\phi < \pi$ is a smooth function of the azimuthal angle $0 \leq \theta \leq 2 \pi$. He next defines $L(\theta) = \int_0^\theta \vert \gamma_\tau'(t) \vert \text{d}t$ to be the length of the curve portion $\gamma_\tau([0,\theta])$ and $A(\theta) = \int_0^\theta (1 - \cos \phi(t)) \text{d}t$ the area on $\Sphere{2}$ swept by the arc of meridian linking the North Pole to the point on $\gamma_\tau$  up to $\theta$. The conformal embedding $f_\tau: \Torus_\tau\to \E^3$ is then given by $f_\tau = f\circ g^{-1}$ with
  \[f: (\R/2\pi\Z)^2\to \E^3, (\theta,\psi)\mapsto \big(\sin\frac{\phi(\theta)}{2}e^{i(\theta+\psi)}, \cos\psi \cos\frac{\phi(\theta)}{2}\big)/(1+\sin\psi \cos\frac{\phi(\theta)}{2}), \,\text{ and}
  \]
  \[g: (\R/2\pi\Z)^2\to \Torus_{-1/\tau}\sim \Torus_\tau, (\theta,\psi)\mapsto (\frac{L(\theta)}{2}, \frac{A(\theta)}{2}+\psi).
  \]
  In other words,
  \[f_\tau(u,v)= \big(\sin\frac{\phi(\theta)}{2}e^{i(\theta+u-A(\theta)/2)}, \cos (u-A(\theta)/2) \cos\frac{\phi(\theta)}{2}\big)/(1+\sin (u-A(\theta)/2) \cos\frac{\phi(\theta)}{2}),
  \]
  where $\theta$ satisfies $L(\theta)=2u$ and $(u,v)\in \R^2/(\Z 2\pi i + \Z 2\pi i\bar{\tau})$.
    
    We have chosen $\phi$ of the form $\phi(\theta) = a + b \sin (n \theta)$ for $a < b$, $0 \leq b < \pi - a$ and $n \in \mathbb{N}$. In order to represent the modulus $\tau=\tau_1+i\tau_i$, the parameters $a,b,n$ should satisfy $A(2\pi) = 4\pi\tau_1$ and $L(2\pi) = 4\pi\tau_i$, or equivalently:
  \[J_0(b) \cos(a) = 1 - 2 \tau_1 \quad \text{and}  \quad \int_0^{2 \pi} \sqrt{n^2 b^2 \cos^2 (n t) + \sin^2 (a + b \sin (n t))} \text{d}t = 4 \pi \tau_i,
  \]
  where $J_0(b) = \frac{1}{\pi} \int_0^\pi \cos(b \sin t) \text{d}t$ denotes the 0-th Bessel function of the first kind. The condition on the total area implies $0\leq \tau_1 \leq 1$. Nevertheless, it is still possible to obtain a conformal embedding in the case of $\tau_1 < 0$ by first reflecting the torus along one of its boundary edge and applying a reflexion of the image torus in $\E^3$. We can thus cover the whole moduli space.
  \subsection{The final construction}
  We now have all the pieces to produce PL isometric embedding of flat tori. Given a modulus $\tau$, we first compute a quasi-equilateral triangulation of $\Torus_\tau$ as in Section~\ref{subsec:acute-triang}. We then compute a PL approximation $F_\tau$ of the conformal map $f_\tau$ defined in Section~\ref{subsec:conformal} and finally apply the construction in Section~\ref{subsec:base-case} to every pair of triangles $(T,F_\tau(T))$. Figures~\ref{fig:square-torus} and~\ref{fig:iso-embeddings} show some results\footnote{More pictures are available on the IMAGINARY website \url{https://www.imaginary.org/gallery/polyhedral-realizations-of-flat-tori}.}.
  \begin{figure}[h]
    \centering
\includegraphics[angle=90,width=.33\textwidth]{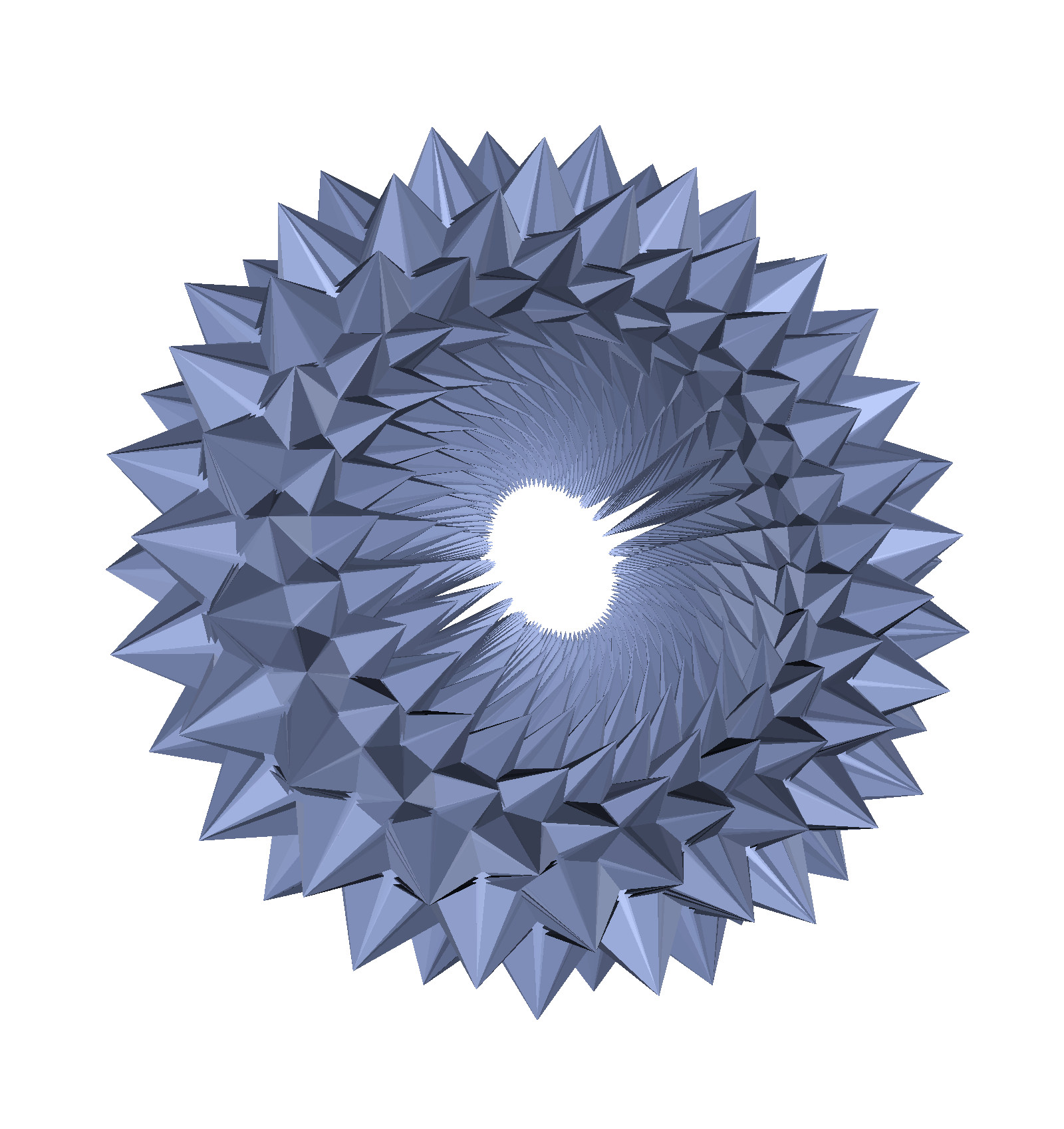}    \includegraphics[angle=90,width=.33\textwidth]{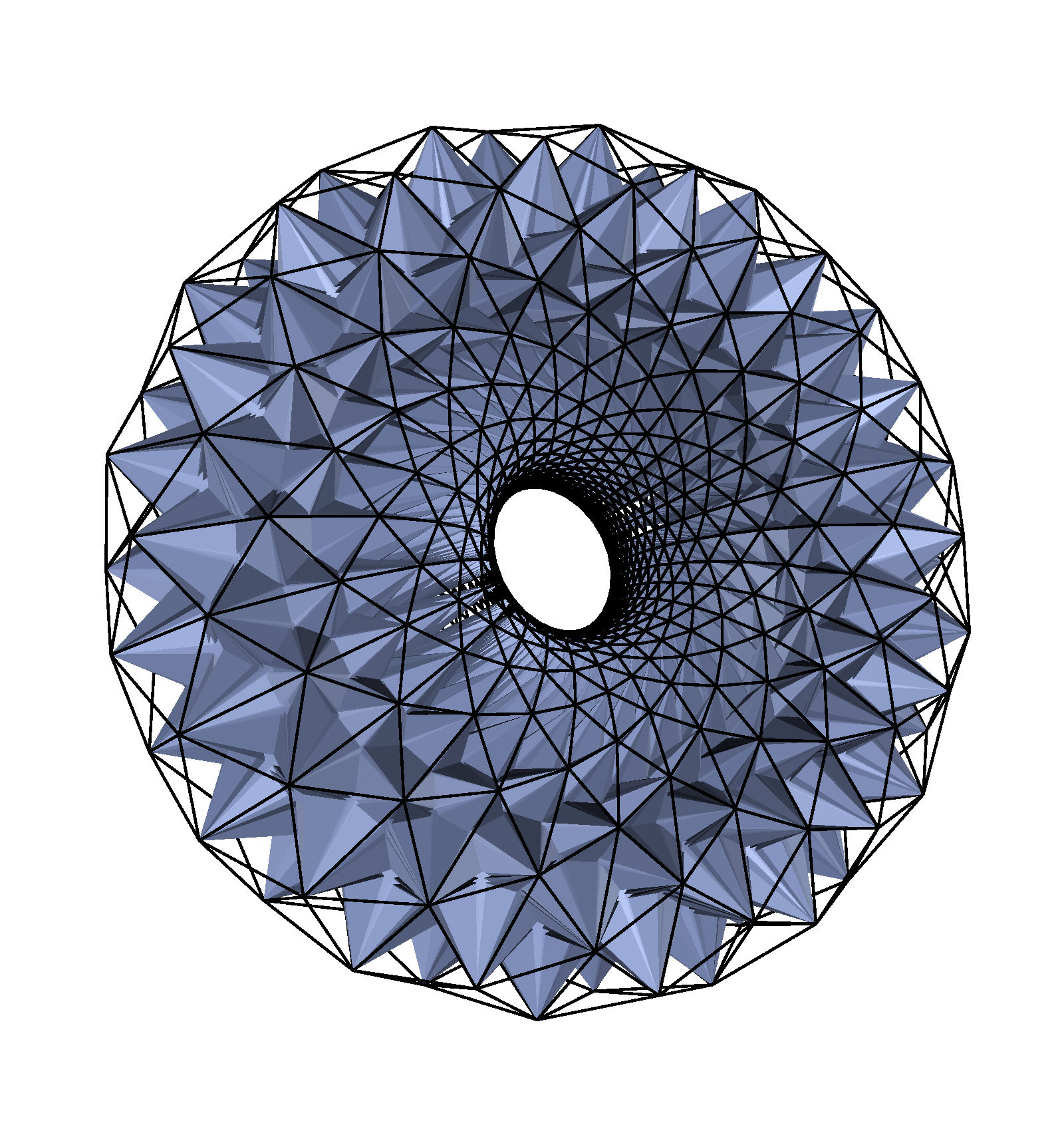}%
\includegraphics[angle=90,width=.33\textwidth]{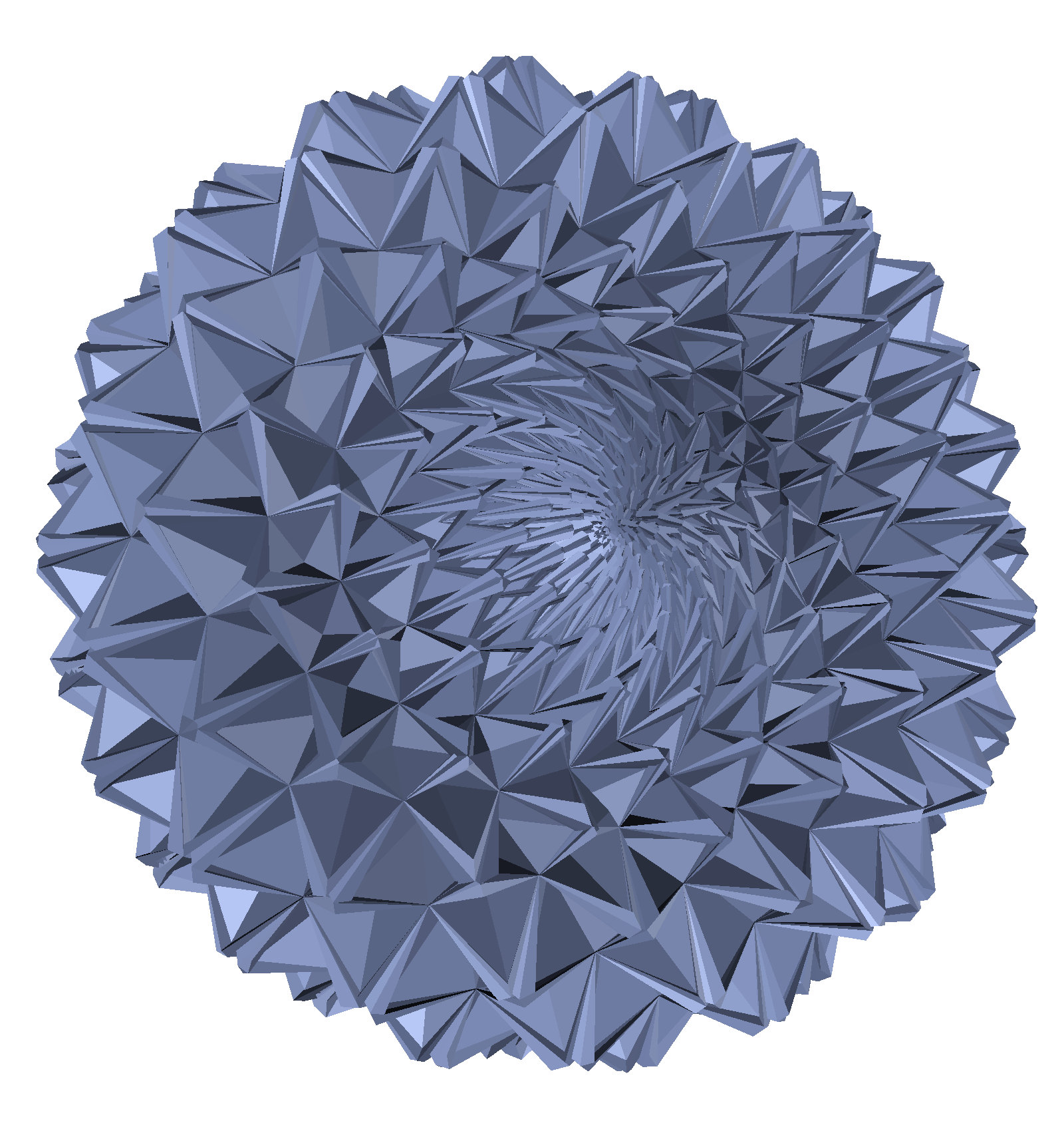}        \caption{Left, PL isometric embedding of the square flat torus with 170,040 triangles. Middle, the mesh with black edges shows the PL approximation of the initial almost conformal embedding. Each of its triangles is replaced with a construction (in blue) as in Section~\ref{subsec:base-case} oriented toward the interior of the initial embedding. Right, The construction is oriented towards the outside, giving another isometric immersion of the square torus. (This last model has self-intersections. A finer triangulation should be used to avoid them.)}
    \label{fig:square-torus}
  \end{figure}
  \begin{figure}[h]
    \centering
\includegraphics[width=.33\textwidth]{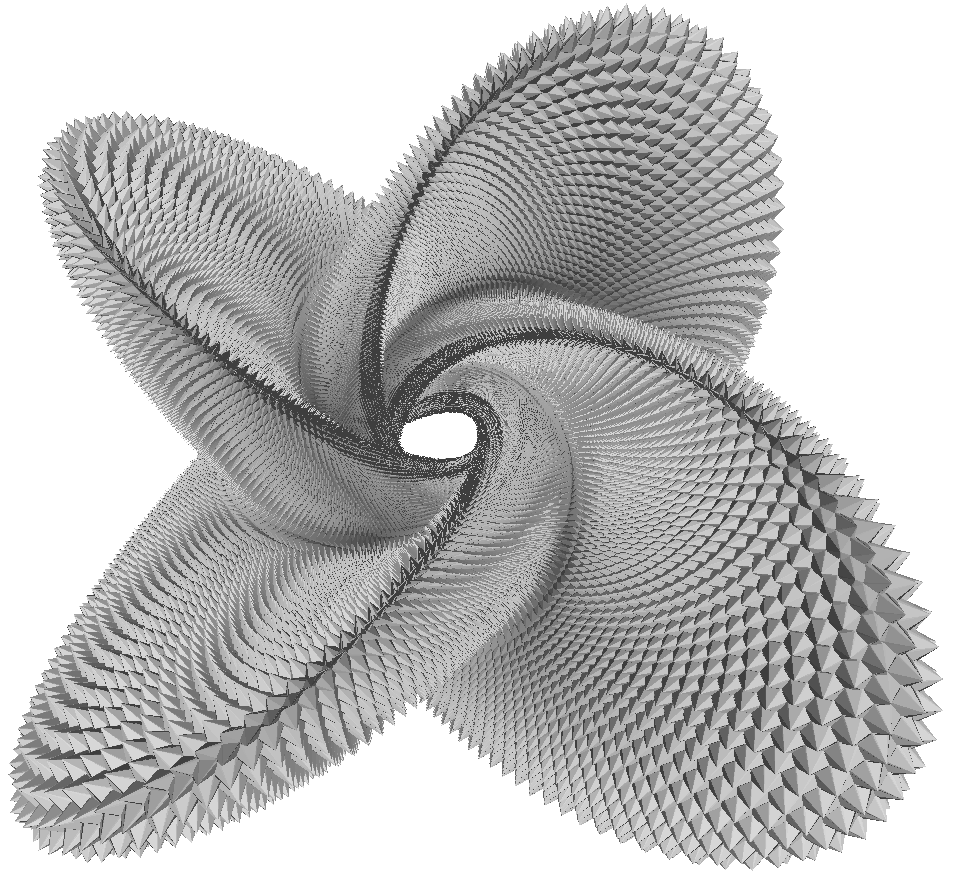}    \includegraphics[width=.33\textwidth]{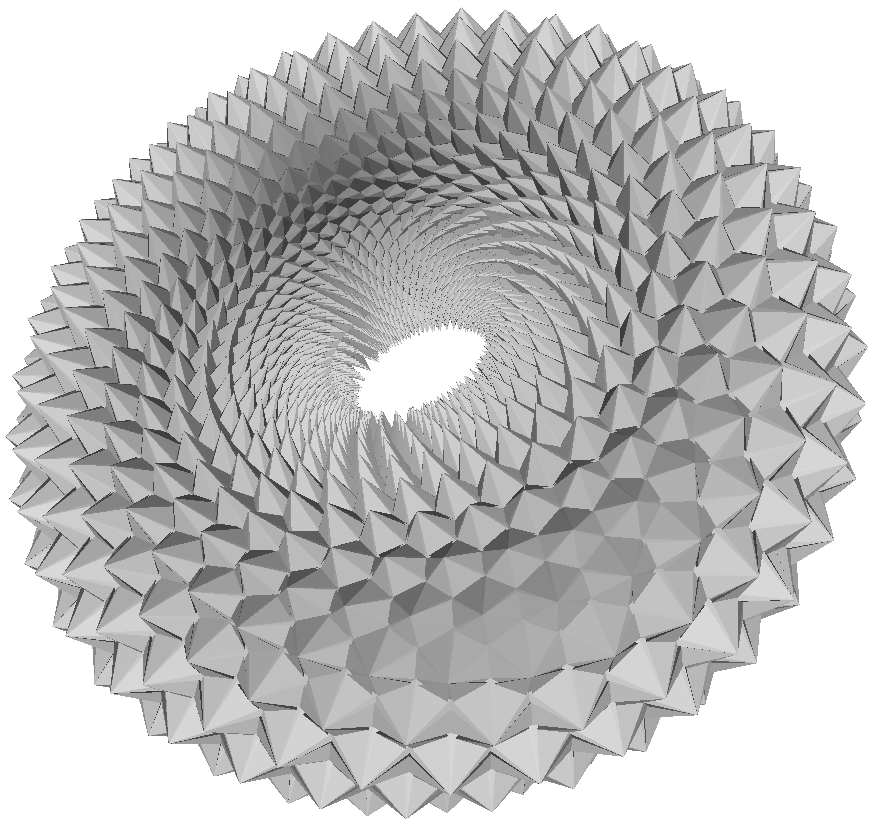}%
\includegraphics[width=.33\textwidth]{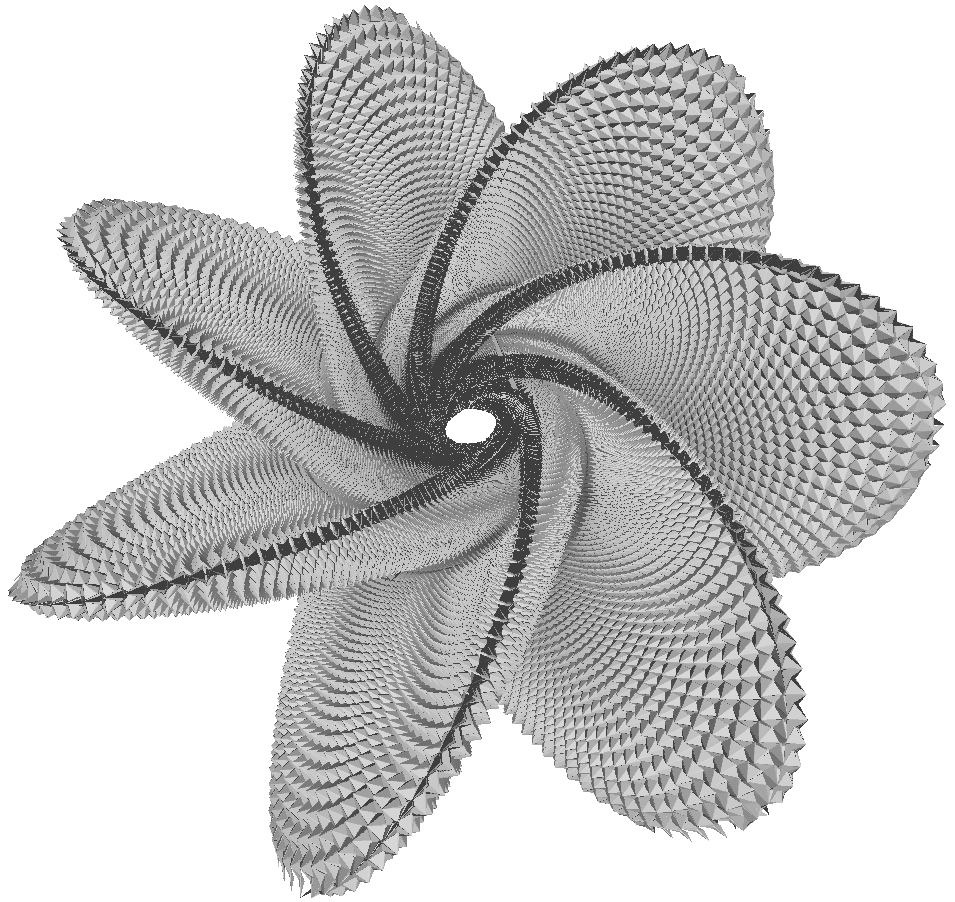}        \caption{isometric immersion of $\Torus_\tau$ with, from left to right, $\tau = e^{i\pi/3}, (1+i)/2, (1+3i)/2$. The left immersion is a hexagonal torus. While the subdivisions of the left and right tori already contain more than $7$ millions triangles, they present self-intersections. A finer triangulation should be used to get an embedding.}
    \label{fig:iso-embeddings}
  \end{figure}

  \section{Universal triangulation}\label{sec:universal-triangulation}
  The construction of Burago and Zalgaller gives rise to triangulations with a huge number of triangles, moreover distinct for every flat torus. In order to get a unique abstract triangulation that admits linear embeddings in $\E^3$ isometric to \emph{any} flat torus, we resort to a second construction by Zalgaller~\cite{z-sblc-00} and to recent works by Tsuboi~\cite{t-oeft-20} and Arnoux et al.~\cite{alm-iplef-21} for embedding flat tori.

  \subsection{Embedding long tori}\label{subsec:long-tori}
  In order to qualify a torus as long or short we define its \emph{aspect ratio} as the ratio of its area by the square of the length of its shortest closed geodesic. Note that this definition is scale invariant. An alternative definition of the aspect ratio uses the modulus of a torus. As recalled in Section~\ref{sec:background}, every flat torus has a modulus $\tau=\tau_1+i\tau_i$ in the fundamental domain shown in Figure~\ref{fig:moduli-space}. Noting that for a torus with such a modulus, its area is equal to $\tau_i$ and its shortest closed geodesic has length 1, we infer that the aspect ratio is given by the imaginary part of the modulus.

Now, a torus is called \define{long} when its aspect ratio is large. In~\cite{z-sblc-00}, Zalgaller proposes an origami style folding of long flat tori, much simpler that the general construction of~\cite{bz-iplit-95}. Here, we quantify how long should be a torus to allow for the Zalgaller folding, and we show that the long tori admit a universal triangulation.
  \begin{proposition}\label{prop:long-tori}
    There exists an abstract triangulation with \TlongT{} triangles, which admits linear embeddings isometric to every torus with aspect ratio larger or equal to 33. 
  \end{proposition}
  The proof essentially follows the construction of Zalgaller adding quantitative bounds to the various steps in the construction. Zalgaller relies on the fact that any flat torus can be obtained from a right  circular cylinder by identifying abstractly its top and bottom boundaries. The non-rectangular tori are obtained by  shifting circularly the top boundary before identification.  Instead of a circular cylinder, Zalgaller starts with a polyhedral cylinder in $\E^3$, namely a prism with equilateral triangular basis, that he bends at several places to make the boundaries coincide, allowing their geometric identification. We show in Lemma~\ref{lem:bending} below that only one of two cases described by Zalgaller for his bending procedure is sufficient for our needs. This leads to a bending with a fixed number of foldings.
A twist is also applied before the bending so as to simulate a circular shift of the top boundary. In general, except for a twist of angle $2k\pi/3$, one boundary will be rotated with respect to the other after the twisting and bending, preventing their identification. Zalgaller then introduces a third device that he calls a \define{gasket} in order to rotate a cross section of the prism without rotating the ``material'' of the prism. Intuitively, one should imagine a sleeve made of some non-elastic fabric, closed by two rigid triangles at the extremities. The right prism results from pulling tight on the triangles. Now, the effect of a gasket is to rotate one triangle around the axis of the prism, allowing the fabric to \emph{slide} along the edges of this triangle.
\paragraph{How to bend a triangular prism}
Consider a right prism $\cal P$ with equilateral triangular basis and an orthogonal cross section $CC'D$.
A \define{bending at an angle $\varphi$ with cutting angle $\lambda$} along the \define{rib} $CC'$ is obtained by (refer to Figure~\ref{fig:bending})
\begin{enumerate}[label=\bfseries(\alph*)]
\item cutting two isosceles triangles $ACB$ and $AC'B$ out of $\cal P$, where $A, B$ lie on the generatrix of the prism through $D$, and the angle at $C$ (and $C'$) is $2\lambda$,
\item bending the cut prism at angle  $0<\varphi<\pi$,
  \item\label{it:fold} folding $ACB$ and $AC'B$ appropriately to fit them back on the bended prism.
\end{enumerate}
\begin{figure}[h]
  \centering
  \includesvg[\textwidth]{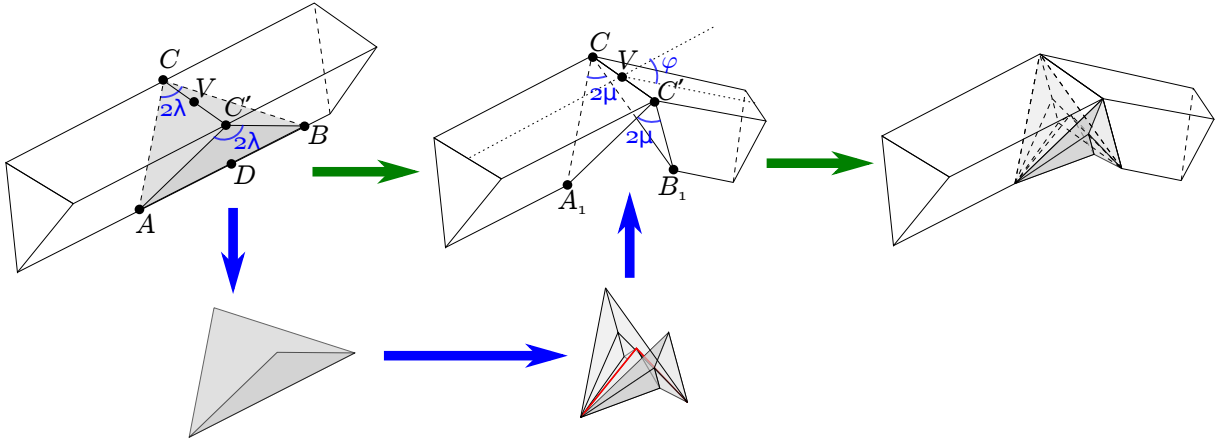}
  \caption{Bending of a prism.}
  \label{fig:bending}
\end{figure}
Let $A_1, B_1$ be the respective positions of points $A, B$ after bending and let $\angle{A_1CB_1}=2\mu$. 
In order for the construction not to overlap, one should have $\mu> 0$, hence $\lambda$ should satisfy $\lambda_0(\varphi) < \lambda < \frac{\pi}{2}$ where $\lambda_0(\varphi)$ is the angle for which, after bending, the triangles $A_1CC'$ and $B_1CC'$ coincide. Looking at the right angled triangles $ADC$ and $ADV$, one easily computes\footnote{This expression gives the same value as the formula given in~\cite{z-sblc-00}, but is somewhat simpler.}
\begin{align}
  \lambda_0(\varphi) = \arctan (\frac{\sqrt{3}}{2}\tan\frac{\varphi}{2}).
  \label{eq:mabda0}
\end{align}

\begin{lemma}\label{lem:bending}
  For every $\varphi\in (0,\pi)$ and for every  $\lambda\in (\lambda_0(\varphi), \pi/2)$, there is an embedded bending of $\cal P$ at angle $\varphi$ with cutting angle $\lambda$ introducing 12 triangles.
\end{lemma}
\begin{proof}
  We only have to check that $ACB$ and $AC'B$ can be folded appropriately. The folding of $ACB$ (and of $AC'B$) is very similar to the folding of a subtriangle as in Figure~\ref{fig:folding-2}. We first fold $CAB$ along its altitude from $C$ to reduce the angle $(CA,CB)$ from $2\lambda$ to $2\mu$. The side $AB$ is mapped onto a broken line $A_1\tilde{D}B_1$ and the goal is to rotate this broken line so that it is contained in the ``vertical'' plane through $A_1,B_1,V$, where $V$ is the middle of $CC'$. See Figure~\ref{fig:bending} and~\ref{fig:preli_2}.
  \begin{figure}[h]
    \centering
    \includegraphics[width=\textwidth]{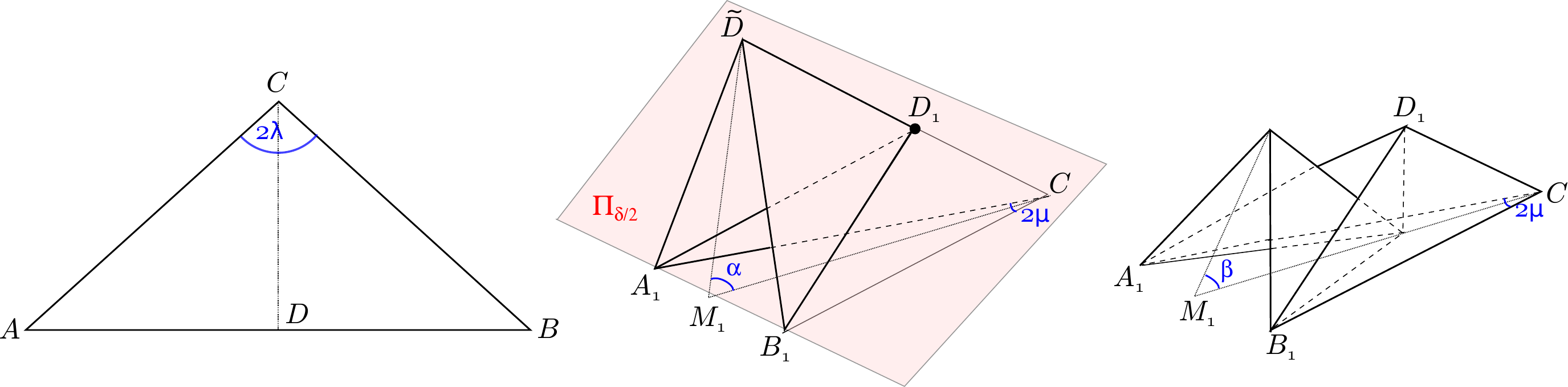}
    \caption{Folding of triangle $ACB$.}
    \label{fig:preli_2}
  \end{figure}
  Denote by $\Pi_0$ the plane spanned by $A_1, B_1, C$. Also denote by $\Pi_\theta$ the plane, in the pencil of planes through $A_1B_1$, making an angle $\theta$ with $\Pi_0$. Let $\alpha,\beta$ be such that  $\Pi_{\alpha}$ contains $\tilde{D}$ and $\Pi_{\beta}$ contains $V$. Because the triangles $C\tilde{D}M_1$ and $CVM_1$, where $M_1$ is the middle of $A_1B_1$, are right angled at $\tilde{D}$ and $V$, we easily deduce\footnote{Our computations do not lead to the same formulas as in~\cite{z-sblc-00}.}
  \[ \alpha = \arcsin\frac{\cos\lambda}{\cos\mu} \quad \text{and}\quad \beta= \arcsin\frac{\cos\lambda}{2\cos\mu}.
    \]
    Set $\delta=\alpha-\beta$. The plane $\Pi_{\delta/2}$ cuts $C\tilde{D}$ in $D_1$. We first reflect across $\Pi_{\delta/2}$ the pieces of triangles $A_1C\tilde{D}$ and $B_1C\tilde{D}$ lying above $\Pi_{\delta/2}$. We next reflect across $\Pi_0$ the part of the reflected part lying below $\Pi_0$. As a result, $A_1\tilde{D}B_1$ is rotated in $\Pi_{\beta}$ as desired. The resulting folding is composed of 6 triangles as illustrated on Figure~\ref{fig:preli_2}.
    Moreover, if $\beta>\delta/2$, then the resulting folding of $ACB$ lies inside the ``top'' quadrant delimited by $\Pi_\beta$ and $\Pi_0$. As a consequence, we can fold $AC'B$ according to the symmetric image across $\Pi_\beta$ of the folding of $ACB$; the two triangle foldings join along the folding of $AB$ in $\Pi_\beta$ and they fit inside the tetrahedron $A_1B_1CC'$ without creating intersections with the rest of the bended prism. They comprise 12 triangles in total.

    In~\cite[Sec. 2]{z-sblc-00}, Zalgaller also considers the case $\beta\leq\delta/2$ that forces him to use an a priori unbounded number of triangles for bending a prism. We claim that there is no need for this second case and that we indeed have $\beta>\delta/2$ for every $\varphi\in (0,\pi)$ and for every  $\lambda\in (\lambda_0(\varphi), \pi/2)$. This inequality is equivalent to $F(\frac{\cos\lambda}{\cos\mu})>0$ for $F(x)=3\arcsin\frac{x}{2}-\arcsin x$. We have $F(0)=F(1)=0$ and a simple computation shows that the derivative of $F$ cancels only once on $[0,1]$ at $x=\sqrt{5/8}$. Since $F(\sqrt{5/8})>0$, we infer that $F$ is positive on $(0,1)$. We finally remark that $\mu<\lambda$, implying $\frac{\cos\lambda}{\cos\mu}\in(0,1)$, which allows us to conclude the claim.
 \end{proof}
  For further reference, we call a \define{bend} a bent prism cut along the orthogonal cross sections through $A_1$ and $B_1$ in the above construction. See Figure~\ref{fig:bend}. 
  \begin{figure}[h]
    \centering
    \includesvg[.6\textwidth]{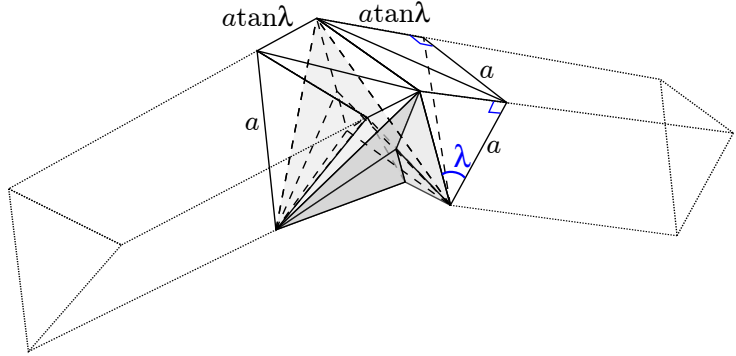}
    \caption{A bend is isometric to a right prism of length $2a\tan\lambda$. It comprises 20 triangles.}
     \label{fig:bend}
  \end{figure}
  After triangulating the two top quadrilaterals, a bend is made of 20 triangles including the 12 triangles as in Lemma~\ref{lem:bending}. 
  
\paragraph{Rotating a cross section with a gasket}
The ribs of the prism $\cal P$ may have only three possible directions corresponding to the three faces of the prism. This prevents to bend $\cal P$ in an arbitrary direction. To circumvent this rigidity, Zalgaller introduces a simple construction that he calls a gasket. Consider an equilateral triangle $ABC$ in the horizontal plane and a vertical translate $A'B'C'$ by height $h$. Rotate $A'B'C'$ by an angle $\alpha$ about
the central vertical axis. The \define{gasket with turn $\alpha$ and height $h$} is the polyhedral cylinder formed by the six congruent triangles $ABA'$, $A'BB'$, $B'BC$, $B'CC'$, $C'CA$, $AA'C'$. See Figure~\ref{fig:gasket}.
  \begin{figure}[h]
    \centering
    \includegraphics[width=.6\textwidth]{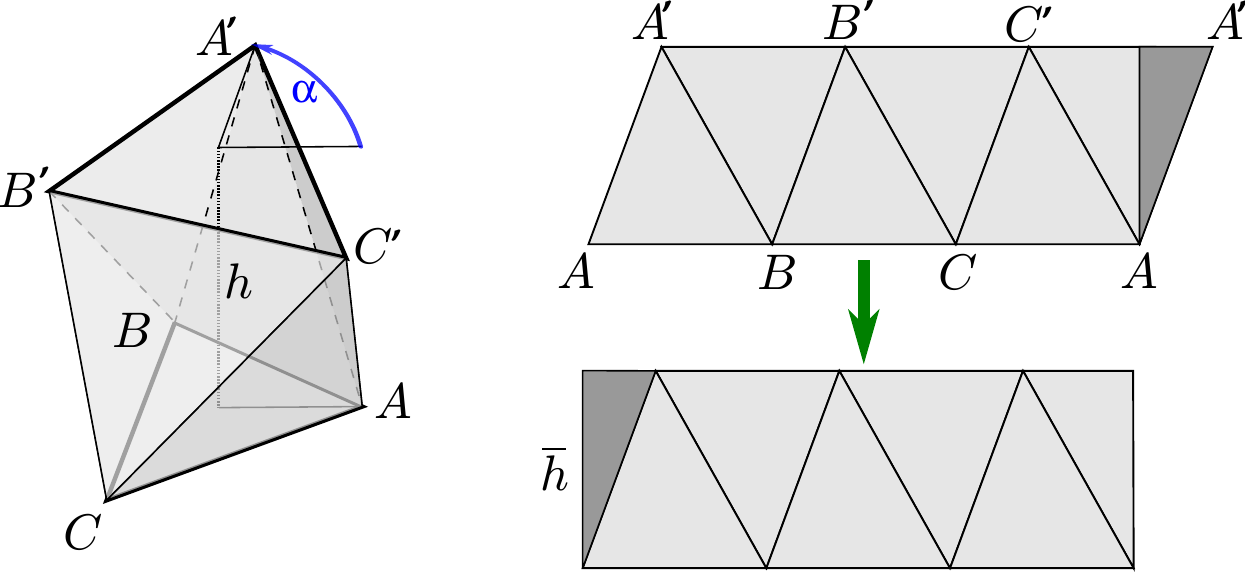}
    \caption{Left, a gasket with turn $\alpha$ and height $h$. Right, the gasket is unfolded in the plane. Cutting and pasting a small triangular piece shows that the gasket has the geometry of a right prism.}
    \label{fig:gasket}
  \end{figure}
  This gasket is embedded for every $\alpha\in (-\pi/3,\pi)$, independently of $h>0$.
  \begin{lemma}\label{lem:gasket}
    For every $\alpha\in (-\pi/3,\pi)$, the gasket with turn $\alpha$ and height $h$ is isometric to a right prism of length $\bar{h}$ with
   \begin{align}
     \bar{h}^2&= h^2 + \frac{2}{27}(\sin^2\frac{\alpha}{2} + \sin^2(\frac{\pi}{3}-\frac{\alpha}{2})) - \frac{4}{81}(\sin^2\frac{\alpha}{2} - \sin^2(\frac{\pi}{3}-\frac{\alpha}{2}))^2 - \frac{1}{36} \label{eq:gasket-length}\\
     &< h^2+\frac{1}{9}.\label{eq:gasket-length-1}
   \end{align}
  \end{lemma}
  \begin{proof}
    By unfolding the gasket in the plane, it is seen to be isometric to a Euclidean rectangle after identifying its vertical sides. See Figure~\ref{fig:gasket}. It is thus isometric to a right prism of length $\bar{h}$, where $\bar{h}$ is the height of the rectangle. Fix the coordinates of $A, B, C$ to respectively, $(\frac{1}{3\sqrt{3}},0)$, $(\frac{e^{i2\pi/3}}{3\sqrt{3}},0)$ and $(\frac{e^{-i2\pi/3}}{3\sqrt{3}},0)$ in $\R^3=\C\times\R$.
    Then the coordinates of $A', B', C'$ are respectively $(\frac{e^{i\alpha}}{3\sqrt{3}},h)$, $(\frac{e^{i(\alpha+2\pi/3)}}{3\sqrt{3}},h)$ and $(\frac{e^{i(\alpha-2\pi/3)}}{3\sqrt{3}},h)$. The three sides of the congruent triangles are given by
    \[|AB| = 1/3, \quad |A'A| =\sqrt{ \frac{4}{27}\sin^2(\alpha/2)+h^2}, \quad |BA'| = \sqrt{\frac{4}{27}\sin^2(\pi/3-\alpha/2)+h^2}.
      \]
 From Heron's formula~\cite{l-s3dph-21}, we have $\bar{h}=2\frac{\sqrt{p(p-|AB|)(p-|BA'|)(p-|A'A|)}}{|AB|}$ with $p$ the half-perimeter of $ABA'$, and we obtain
 \begin{align*}
   \bar{h}^2=h^2 + \frac{2}{27}(\sin^2\frac{\alpha}{2} + \sin^2(\frac{\pi}{3}-\frac{\alpha}{2})) - \frac{4}{81}(\sin^2\frac{\alpha}{2} - \sin^2(\frac{\pi}{3}-\frac{\alpha}{2}))^2 - \frac{1}{36}.
 \end{align*}
Since
  \[\sin^2\frac{\alpha}{2} + \sin^2(\frac{\pi}{3}-\frac{\alpha}{2}) = \frac{3}{4} -\sin\frac{\alpha}{2}\sin(\frac{\pi}{3}-\frac{\alpha}{2}) < 7/4,
  \]
  it follows that
  \begin{align*}
    \bar{h}^2< h^2+\frac{2}{27}\times\frac{7}{4}-\frac{1}{36} < h^2+\frac{1}{9}.
 \end{align*}
  \end{proof}
  By pasting two prisms at the boundaries of a gasket, we obtain a polyhedral cylinder with triangular boundaries, where the two boundaries are turned at the angle $\alpha$ with respect to each other; see Figure~\ref{fig:gasket-1}.
  \begin{figure}[h]
    \centering
    \includegraphics[width=.5\textwidth]{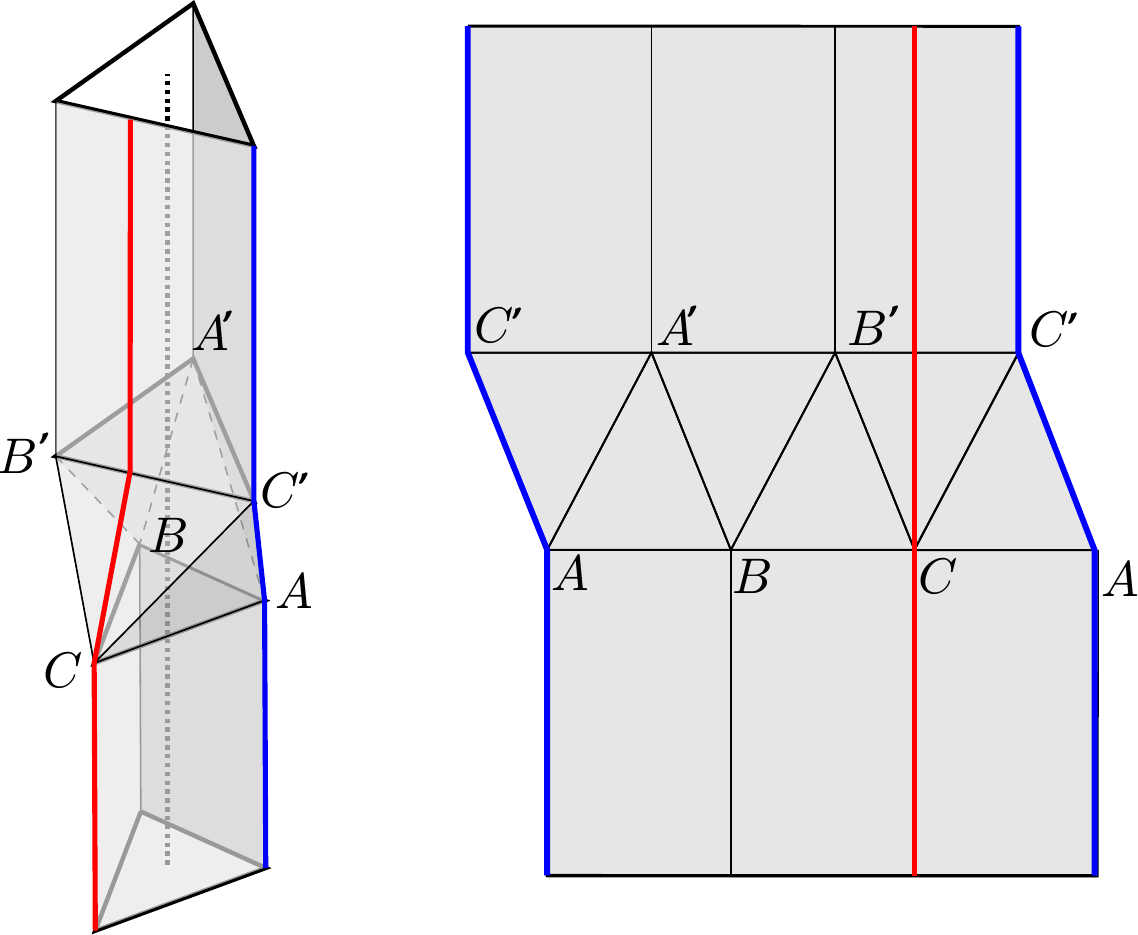}
    \caption{Joining a gasket with two prisms to rotate their ribs. Right, unfolding of the construction showing the line of cut (in blue) and a generatrix (in red) of the polyhedral cylinder.}
    \label{fig:gasket-1}
  \end{figure}
  \begin{note}
   The top and bottom prisms in Figure~\ref{fig:gasket-1} have the same central axis.
  This allows to rotate the rib of a prism at an angle $\alpha\in(-\pi/3,\pi)$ before applying a bending.
 \end{note}
 \begin{note}\label{nt:gasket}
   By joining $k$ gaskets in a row, we can rotate the rib of a prism at an angle  $\alpha\in(-k\pi/3,k\pi)$. 
 \end{note}
 \paragraph{Twisting a prism}
 Replacing a portion of a prism by a gasket with turn $\alpha$ allows to turn one boundary, say the top one, of the prism with respect to the other one but does not \emph{twist} the prism: the gasket makes generatrices going
through a vertex in the bottom triangle not go through the vertex in the
top. In order to twist the prism so that the top endpoint of this geodesic indeed turns with the boundary, Zalgaller introduces yet another construction that he calls a \define{helical twist}. This construction takes advantage of the holonomy of parallel transport on the sphere: consider a unit sphere of center $O$ with a spherical triangle $PQR$ (see Figure \ref{fig:holonomy}).
 \begin{figure}[h]
   \centering
   \includegraphics[width=.3\textwidth]{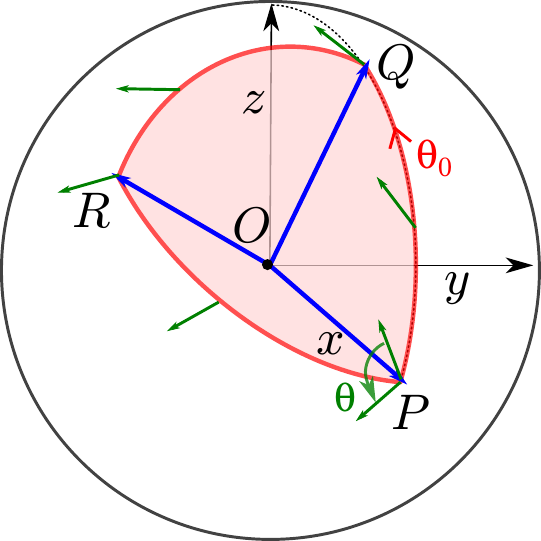}
   \caption{The green tangent vector is transported along the spherical triangle $PQR$. The angle $\theta$ is equal to the area of $PQR$, while the angle $\theta_0$ is given by L'Huillier's formula.}
   \label{fig:holonomy}
 \end{figure}
 If one parallel transports an object from $P$ to $P$ following the sides of the triangle $PQR$, then the object is rotated by a certain angle around the axis $OP$ which is equal to the signed area of the spherical triangle $PQR$. In order to twist a prism with axis directed by $\overrightarrow{OP}$ by an angle $\theta$ we may thus bend the prism successively in the directions $\overrightarrow{OQ}, \overrightarrow{OR}$ and $\overrightarrow{OP}$ again. Each bending at angle $\varphi$ indeed corresponds to a transport along a spherical geodesic of length $\varphi$. Each portion of prism between two bends should include two gaskets to orient its rib properly. Indeed, by Note~\ref{nt:gasket}, two gaskets allow to turn by an angle in $(-2\pi/3,2\pi)$, which covers all the possible orientations\footnote{While the symmetry of the cross section would allow to turn by an angle with amplitude at most $\pi/3$, it is in fact necessary to allow for all the possible angles in order to obtain a unique isomorphism class of triangulations after the final gluing.}.

 A \define{helical twist of angle $\theta$} consists of a sequence of gaskets and bends according to the  pattern $(g^2b)^5g^2=(g^2b)^3(g^2b)^2g^2$, where $b,g$ stand respectively for bends and gaskets. The prefix $(g^2b)^3$ in the pattern is used to simulate the parallel transport as described above, assuming that the central axis of the initial cross section is already aligned with $\overrightarrow{OP}$. The next factor $(g^2b)^2$ allows to return on the central axis of the initial cross section. Since $\overrightarrow{OP}$ is aligned with this central axis, the  changes of direction due to the factor $(g^2b)^2$ happen in the same plane. The resulting holonomy is thus trivial, which ensures that the first and last cross sections of this portion are parallel. Finally, the last two gaskets allows to turn the cross section by any angle in $(-2\pi/3,2\pi)$; see Figure~\ref{fig:helical-twist}. 
 \begin{figure}[h]
   \centering
   \includegraphics[width=.9\textwidth]{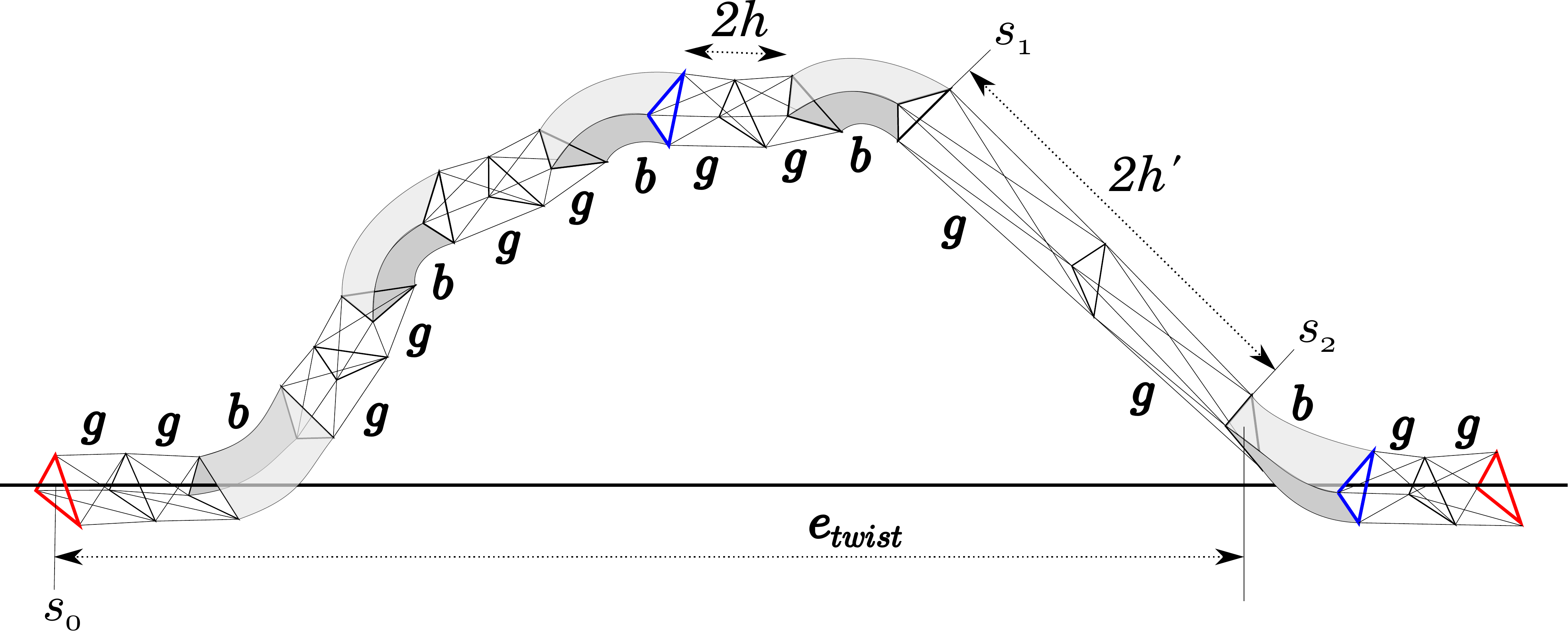}
   \caption{The cross section (in blue) after the last bending of a helical twist is rotated by an angle $\theta$ about the central axis with respect to the initial  cross section (in red). The last two gaskets allow to turn the last cross section (in red) to be a translate of the first one.}
   \label{fig:helical-twist}
 \end{figure}

 In practice, to construct a helical twist of angle $\theta \in (-\pi, \pi]$, we choose an equilateral triangle $PQR$ on the unit sphere, with area $\theta$. Moreover, we fix $P = (1, 0, 0)$, and we take $Q$ in the plane $Oxz$ with positive $z$ coordinate. Then, $R$ is the unique point making $PQR$ equilateral and counterclockwise. Denote by $\theta_0$ the angle between the vectors $\overrightarrow{OP}$ and $\overrightarrow{OQ}$. L'Huillier's formula relates the area $\cal A$ of a geodesic triangle on the unit sphere with its side lengths $a,b,c$ by
 \[\tan \frac{\cal A}{4} = \sqrt {\tan \left({\frac {p}{2}}\right)\tan \left({\frac {p-a}{2}}\right)\tan \left({\frac {p-b}{2}}\right)\tan \left({\frac {p-c}{2}}\right)},
 \]
 where $p$ is the half-perimeter. It follows that $\theta_0$ satisfies the equation
 \[\theta = 4\arctan\left(\sqrt{\tan(\frac{3\theta_0}{4})\tan^3(\frac{\theta_0}{4})}\right).
 \]
 This equation has an explicit solution for $|\theta|\leq \pi$:
 \[\theta_0 = 4\cdot\text{sign}(\theta)\cdot \arctan\left(\sqrt{1+2\frac{\cos(\frac{|\theta|}{12}-\frac{2\pi}{3})}{\cos\frac{|\theta|}{4}}}\right).
     \]
Traveling along $PQR$ in trigonometric direction induces a positive rotation angle, while traveling clockwise induces a negative rotation angle.
For $|\theta|\leq \pi$, the above formula implies
\[\theta_0\leq 4 \arctan \sqrt{1+2\sqrt{2}\cos \frac{7\pi}{12}} = 4 \arctan \sqrt{2 - \sqrt{3}}\approx 1.911.\]
From \eqref{eq:mabda0}, we deduce that the corresponding cutting angle satisfies \[ \lambda_0(\theta_0) \leq \arctan \frac{\sqrt{6-3\sqrt{3}}}{\sqrt{3}-1}< \arctan \frac{49}{40}.
\]
For further reference, we set
\[\lambda_0:=\arctan \frac{49}{40}.
\]
Denote by  $s_0$ the initial cross section of the helical twist, by $s_1$ the cross section at the end of the fourth bend, and by $s_2$ the initial cross section of the last bend. Refer to Figure~\ref{fig:helical-twist}.
 \begin{lemma}\label{lem:helical-twist}
   Given any twist angle $\theta\in (-\pi, \pi]$ and any $h>0$, we can construct a helical twist of angle $\theta$ so that all its bends have cutting angle $\lambda_0$, and all its gaskets have height $h$, except the two gaskets between sections $s_1$ and $s_2$, which have height $h'$ imposed by our construction. This helical twist is isometric to a right prism with length
   \[\ell_{\text{twist}} = 10a\tan\lambda_0+10 \bar{h} + 2\bar{h'}\]
   and the horizontal distance between the boundaries of the helical twist is bounded by
   \[d_{\text{twist}} = 18(a\tan\lambda_0 + h) .
   \]
   Here, $\bar{h}$ and $\bar{h'}$ are given by Equation~\eqref{eq:gasket-length} in Lemma~\ref{lem:gasket}.
   The height $h'$ is moreover bounded by $2\sqrt{10}(2h+3a\tan\lambda_0)$.
 \end{lemma}
 \begin{proof}
   Let $a=1/3$ be the length of a rib, \ie, of a side of the triangular cross-sections. The bending angle of the three first bends is equal to $\theta_0$ and we know by the previous discussion that they can be realized with the cutting angle $\lambda_0$. We need to prove that the last two bends can be realized with this cutting angle. Define the \define{extent $e_{\text{twist}}$} of the helical twist as the horizontal distance between the centers of the sections $s_0$ and $s_2$. We fix $e_{\text{twist}}=16(h+a\tan\lambda_0)$. Let $c_i$ be the center of $s_i$, $i=0,1,2$. The last two bends have the same bending angle $\varphi$, which is the angle between the horizontal direction and the vector $\overrightarrow{c_1c_2}$. We have $\tan\varphi = \delta_v/\delta_h$, where $\delta_v$ is the distance of $c_1$ to the horizontal line through $c_2$ and $\delta_h$ is the horizontal distance between $c_1$ and $c_2$. We estimate $\delta_v$ by adding the contributions of the eight gaskets and the four bends preceding $s_1$. Four of the eight gaskets are horizontal; it follows that they do not contribute to $\delta_v$. The fourth bend turns towards the horizontal axis through $c_2$, it thus contributes negatively to $\delta_v$. We infer that $\delta_v\leq 4h + 6a\tan\lambda_0$. On the other hand, the horizontal distance $d_{01}$ between $c_0$ and $c_1$ is bounded by $8h+8a\tan\lambda_0$. It ensues that $\delta_h=e_{\text{twist}} - d_{01}\geq 8h+8a\tan\lambda_0$. We conclude that $\tan\varphi \leq (4h + 6a\tan\lambda_0)/(8h+8a\tan\lambda_0)$. Hence, for all $h$, $\tan\varphi < 3/4$. Using Equation~\eqref{eq:mabda0} and the classical formula $\tan\varphi = 2\tan\frac{\varphi}{2}/(1-\tan^2\frac{\varphi}{2})$, we deduce $\tan \lambda_0(\varphi) < \frac{1}{2\sqrt{3}}$. It follows that $\lambda_0(\varphi)<\lambda_0$ as desired.

   The \emph{intrinsic length} of the helical twist, \ie, the height of the isometric right cylinder, is the sum of the intrinsic lengths of each constituting bend and gasket. We thus obtain the formula as in the lemma, where $\bar{h}$ and $\bar{h'}$ are given by~\eqref{eq:gasket-length}. We now remark that the total horizontal extent of the helical twist is bounded by $e_{\text{twist}}+2a\tan\lambda_0+2h$ to obtain the bound in the lemma.

   We next remark that $d_{01}\geq 4h+2a\tan\lambda_0$, taking into account the four horizontal gaskets and the 2 incident horizontal half bends. Hence, \\
   $\delta_h = e_{\text{twist}} - d_{01}\leq 12h+14a\tan\lambda_0$ and we finally conclude
   \begin{align*}
     h' = \sqrt{\delta_h^2+\delta_v^2}&\leq \sqrt{(12h+14a\tan\lambda_0)^2 + (4h + 6a\tan\lambda_0)^2}\\
     &\leq 2\sqrt{10}(2h+3a\tan\lambda_0).
   \end{align*}
 \end{proof}

 \paragraph{Putting the pieces together}
 Consider a flat torus with modulus $\tau=\tau_1+i\tau_i$. It can be obtained from a right circular cylinder of height $\tau_i$ and boundary length 1, identifying the boundaries after a circular shift at angle $2\pi\tau_1$. Zalgaller constructs his PL isometric embeddings of long tori, for which $\tau_i$ is large, as follows. He first replaces the circular cylinder by an isometric equilateral triangular prism that is bent 6 times at angle $\pi/3$ to form a hexagonal tube. If the torus is rectangular, that is if $\tau_1=0$, then the identification of the initial and final cross sections provides the desired embedding. Otherwise, he replaces one side of the hexagon by a helical twist of angle $2\pi\tau_1$ in order to glue the boundaries of the prism with the correct angular shift.
We use a slightly different construction that allows us to get shorter tori. Starting from a helical twist of angle $\theta$, we add 4 bends at angle $\pi/2$ and 3 portions of right prisms as illustrated on Figure~\ref{fig:torusConstruct}
 \begin{figure}[h]
   \centering
    \includegraphics[width=.8\textwidth]{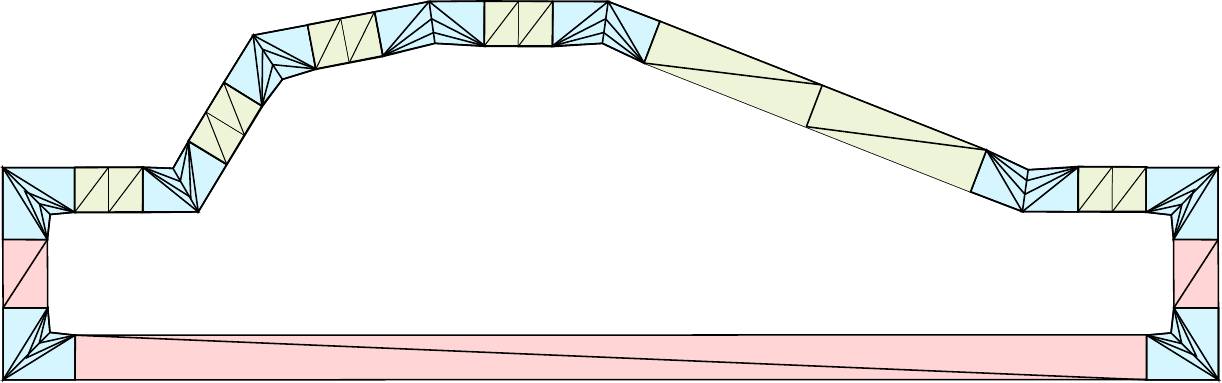}
   \caption{Our construction decomposed into bends (in light blue), gaskets (in light green) and triangular prisms (in pink).}
   \label{fig:torusConstruct}
 \end{figure}
 to form a closed torus. In order to avoid intersections between the horizontal prism and the horizontal  gaskets of the helical twist we choose the two vertical prisms of length $\frac{a}{3}>\frac{a}{2\sqrt{3}}$. We also choose the length of the horizontal prism to be equal to the total horizontal extent of the helical twist. 
 We finally take the cutting angle of the 4 right angled bends equal to $\lambda_0':=\arctan(9/10)>\lambda_0(\pi/2)$. The resulting torus has length
 \begin{align*}
   L < \ell_{\text{twist}} + 8a\tan\lambda_0' + 2a/3 + d_{\text{twist}},
 \end{align*}
 where $\ell_{\text{twist}}$ and $d_{\text{twist}}$ are given by Lemma~\ref{lem:helical-twist}.  In other words,
 \[L < 28a\tan\lambda_0 + 8a\tan\lambda_0' + 2a/3 + 18h + 10\bar{h}+2\bar{h'}.
   \]
 Using the bound for $h'$ in Lemma~\ref{lem:helical-twist} together with inequality~\eqref{eq:gasket-length-1}, and the fact that $\tan\lambda_0= 49/40$, and $\tan\lambda_0' = 9/10$, we get
 \begin{align}\label{eq:L-bound}
 L < \frac{253}{18} + 18h+10\sqrt{h^2+\frac{1}{9}}+2\sqrt{40\left( 2h + \frac{49}{40}\right)^2+\frac{1}{9}}.
\end{align}
By taking $h=0$ we thus obtain $L<\frac{253}{18} +\frac{10}{3}+2\sqrt{\frac{49^2}{40}+\frac{1}{9}}<\Lmin$. Note that any longer torus can be obtained by elongating the two vertical prisms. Hence, for $h$ strictly positive and small enough, say smaller than 0.002, we can realize any flat torus of length at least $\Lmin$. The bound for $L$ in~\eqref{eq:L-bound} is largely overestimated and our implementation shows that the same construction, taking $h=1/15$, allows to embed tori of length much shorter than 33 even though the right member in~\eqref{eq:L-bound} evaluates to more than 33. Some rendering is visible on Figure~\ref{fig:geomview}.
   \begin{figure}[h]
   \centering
   \includegraphics[width=\textwidth]{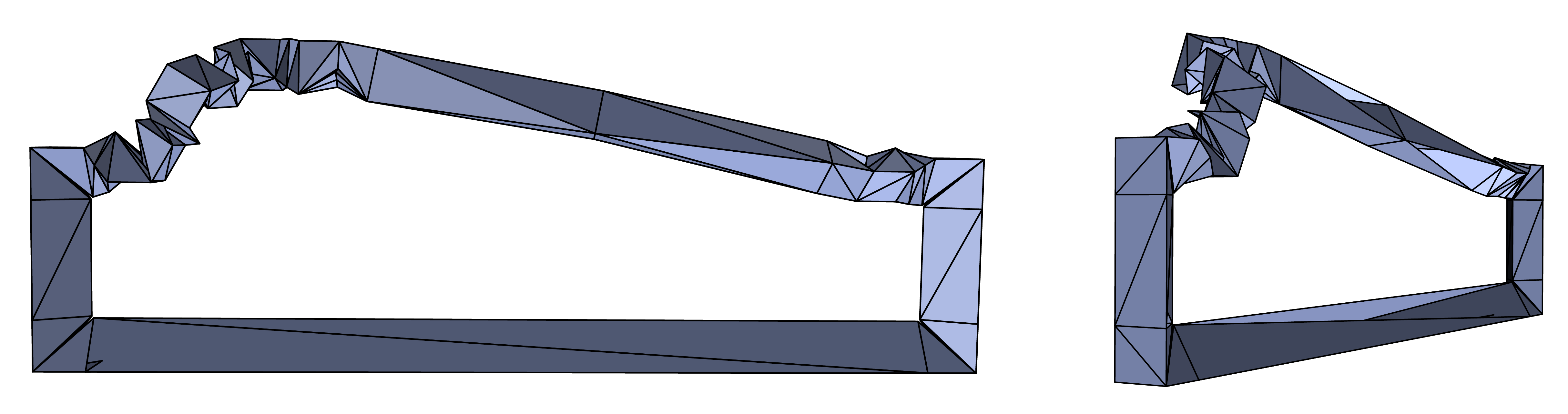}
   \caption{Rendering of the flat torus of length 18 and circular shift $2\pi/5$ corresponding to the modulus $0.2 + 18i$.}
   \label{fig:geomview}
 \end{figure}

 We remark that a prism can be triangulated as a gasket with turn 0, the whole construction thus corresponds to the pattern $(g^2b)^5g^2(bg)^3b$ and is composed of $15\times 6+9\times 20 = \TlongT$ triangles. This ends the proof of Proposition~\ref{prop:long-tori}. Figure~\ref{fig:patronLongTori-ter} shows the resulting unfolded triangulation after cutting through a cross section and a longitude. 
 \begin{figure}[h]
   \centering
   \includegraphics[width=\textwidth]{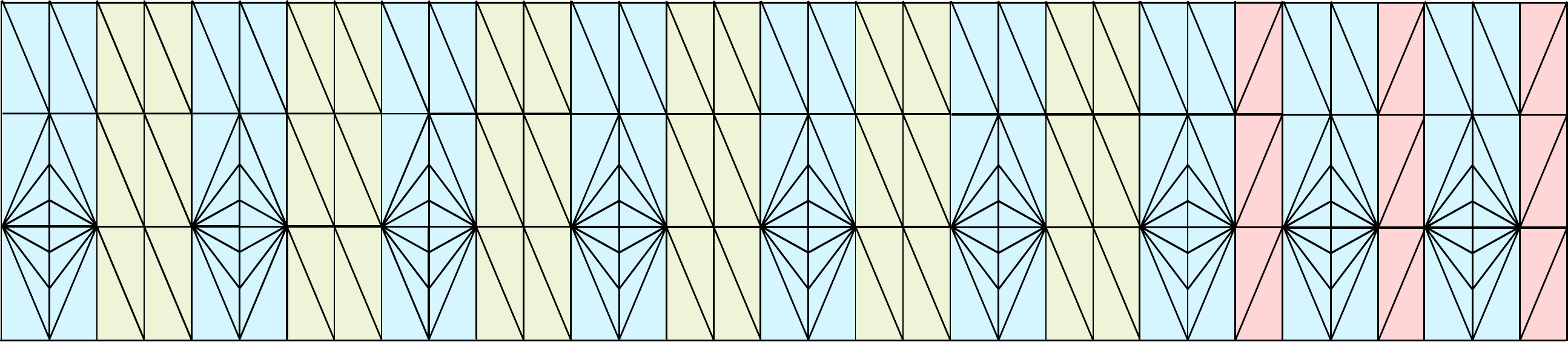}
   \caption{A universal triangulation for long tori.}
   \label{fig:patronLongTori-ter}
 \end{figure}

 \subsection{The diplotori of Tsuboi and Arnoux et al.}\label{subsec:diplotori}
 The previous construction provides a universal triangulation for long tori. Referring to Figure~\ref{fig:moduli-space}, this means that the part of the moduli space above the horizontal line $\tau_i=\Lmin$ can be geometrically realized in $\E^3$ by this unique abstract triangulation. It thus remains to cover the compact subspace of \emph{short tori} below this line. Denote this subspace by ${\cal M}_{\text{short}}$. Hence,
 \[{\cal M}_{\text{short}}=\{\tau\in\Hyp^2\mid |\tau|\geq 1, |\tau_1|\leq 1/2, |\tau_i|\leq \Lmin\}.
   \]
 As already observed in Section~\ref{sec:BZ95}, the construction of Burago and Zalgaller allows for some flexibility, implying that around every point in the moduli space there is a neighborhood that can be geometrically realized by the same abstract triangulation. By compactness we can cover ${\cal M}_{\text{short}}$ with a finite number of such neighborhoods. We could thus overlay all the corresponding triangulations with the universal triangulation for long tori to obtain a universal triangulation for all tori. This already provides a proof of the existence of such a triangulation. However, estimating the size of the neighborhoods seems impractical and the approach would lead to a gigantic triangulation. Surprisingly, it was only very recently that Tsuboi~\cite{t-oeft-20} and Arnoux et al.~\cite{alm-iplef-21} independently (re)discovered extremely simple geometric realizations of flat tori.  (See Segerman~\cite[Fig. 6.12 and Appendix A]{s-vm3dp-16} for some historical clues.) Arnoux et al. are able to prove that their construction, that they call \emph{diplotorus}, allows to realize all tori in the moduli space. For completeness we briefly recall this construction.

 The diplotorus ${\cal D}_{n,d}^{a,h}$ with parameters $n,d,a,h$ is defined as follows.
 Let $A_k=(e^{i\frac{2\pi k}{n}},0)$ be the vertices of the regular $n$-gon in the horizontal coordinate plane. Let $B_k=(e^{i\frac{\pi}{n}(a+1+2k)},h)$ be the vertices of the vertical translate by $h$ of this $n$-gon, turned by an angle $(a+1)\frac{\pi}{n}$.
 Then ${\cal D}_{n,d}^{a,h}
= {\cal P}_{\text{int}}\bigcup {\cal P}_{\text{ext}}$ is the union of two twisted prisms, called \emph{ploids}, where ${\cal P}_{\text{int}}$ is the union of triangles $\{A_kA_{k+1}B_k\}_{0\leq k<n}$ and $\{B_kA_{k+1}B_{k+1}\}_{0\leq k<n}$, and  ${\cal P}_{\text{ext}}$ is the union of triangles $\{A_kA_{k+1}B_{k-d}\}_{0\leq k<n}$, $\{B_{k-d}A_{k+1}B_{k+1-d}\}_{0\leq k<n}$. Of course, all the indices should be considered modulo $n$. (Note that a ploid with $n=3$ is nothing else but a gasket.) Figure~\ref{fig:diplotorus-view} shows the diplotorus ${\cal D}_{5,2}^{3.5,2}$. 
  \begin{figure}[h]
   \centering
   \includegraphics[width=\textwidth]{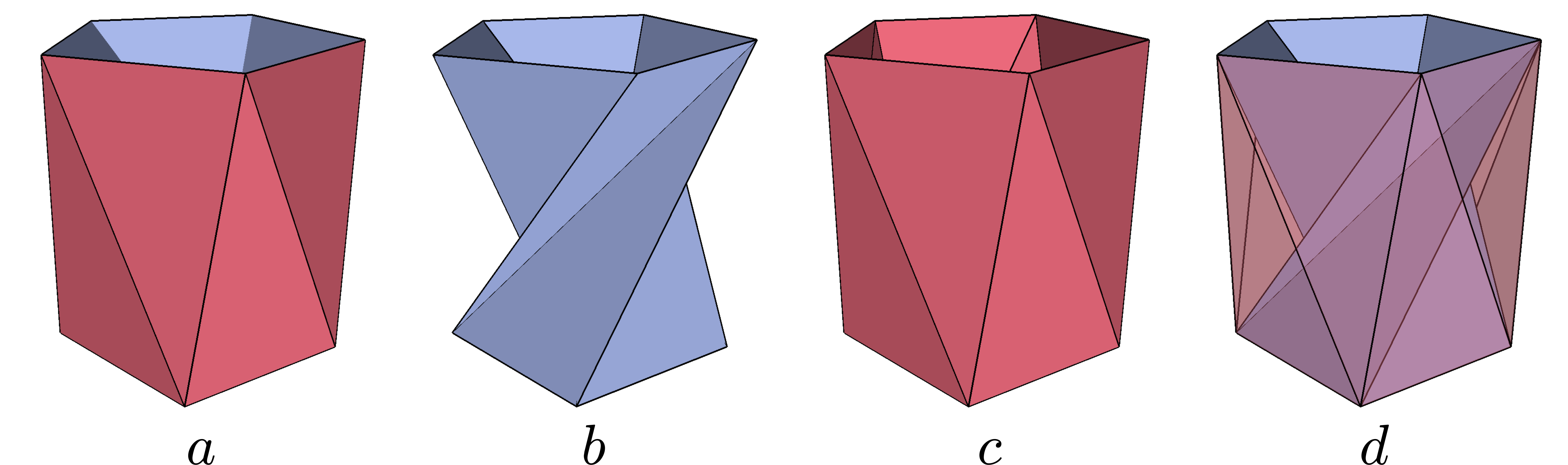}
   \caption{View of the diplotorus ${\cal D}_{5,2}^{3.5,2}$ (a) with its internal (b) and external (c) ploids. (d), another view of ${\cal D}_{5,2}^{3.5,2}$ with a transparent external ploid.}
   \label{fig:diplotorus-view}
 \end{figure}
 \begin{lemma}[Arnoux et al., 2021]\label{lem:Arnouxetal}
   For $h,a\in\R$ and $n,d\in\Z$, ${\cal D}_{n,d,}^{a,h}$ is an embedded flat torus if and only if
   \[h>0, n>4, 2\leq |d|<n-2,  d+1<a<n-1 \text{ if } d>0,\quad\text{and}\quad 1-n<a<d-1 \text{ if } d<0
     \]
     Moreover, the modulus of ${\cal D}_{n,d}^{a,h}$ is $\tau(n,d,a,h) = \tau_1(n,d,a)+i\tau_i(n,d,a,h)$ with
     \begin{align*}
       \tau_1(n,d,a) =& d/n - \frac{\cos((a-d) \frac{\pi}{n}) \sin(d \frac{\pi}{n})}{n \sin\frac{\pi}{n}} \quad\text {and} \\
       \tau_i(n,d,a,h)  =& \Big(\sqrt{h^2+4 \sin^2(\frac{a+1}{2}\cdot\frac{\pi}{n}) \sin^2(\frac{a-1}{2}\cdot\frac{\pi}{n})} \quad + \\
       & \sqrt{h^2+4 \sin^2(\frac{a-2d+1}{2}\cdot\frac{\pi}{n}) \sin^2(\frac{a-2d-1}{2}\cdot\frac{\pi}{n})}\,\,\Big)/(2 n \sin(\pi/n))
     \end{align*}
 \end{lemma}
 The map $\tau_1$ does not depend on $h$ while $\tau_i$ is an increasing function of $h$. It follows that for $n$ and $d$ fixed, the moduli of the tori ${\cal D}_{n,d}^{a,h}$ form a subset of $\Hyp^2$, which we denote by  ${\cal M}_{n,d}$, that lies above the graph of the parametrized curve $a\mapsto (\tau_1(n,d,a), \tau_i(n,d,a,0))$, where $a$ varies as in Lemma~\ref{lem:Arnouxetal}. For $n$ and $d$ fixed, the diplotori in ${\cal D}_{n,d}^{a,h}$ have the same abstract triangulation. If one could cover the moduli space with a finite number of regions ${\cal M}_{n,d}$, this would therefore provide a universal triangulation. This is however impossible and one needs to let $n$ grow to infinity to realize all the rectangular tori. We can nonetheless cover the moduli of short tori with only three regions ${\cal M}_{n,d}$.

\subsection{Embedding short tori with three diplotori}\label{subsec:short-tori}
The fundamental domain in Figure~\ref{fig:moduli-space} is symmetric with respect to the imaginary axis. Two symmetric points $\tau$ and $-\bar{\tau}$ actually represent isometric tori, but the isometry should reverse the orientation. Hence, if $\Torus_\tau$ has a PL isometric embedding in $\E^3$ so does $\Torus_{-\bar{\tau}}$: just take a reflected image of the embedding of $\Torus_\tau$. It is thus enough to realize the positive part ${\cal M}_{\text{short}}^+ := \{\tau\in {\cal M}_{\text{short}}\mid \tau_1\geq 0\}$ of the short tori to ensure that we can realize all of them.
\begin{remark}
  From Section~\ref{sec:background} the moduli space is the quotient of $\Hyp^2$ by the action of $\text{SL}_2(\Z)$. To realize all the short tori, it is thus sufficient to realize their moduli in any image of ${\cal M}_{\text{short}}^+$  under the action of $\text{SL}_2(\Z)$. Figure~\ref{fig:ModularGroup-FundamentalDomain} shows such an image. The top side $\tau_i=\Lmin$ of ${\cal M}_{\text{short}}^+$ is transformed to an arc of circle, called a \emph{horocycle}, tangent  at 0 to the real axis. In Figure~\ref{fig:ModularGroup-FundamentalDomain}, the blue and yellow lines are part of the \emph{Dedekind tessellation} of the hyperbolic plane. It tiles $\Hyp^2$ into hyperbolic triangles with one ideal vertex. Each such triangle is a fundamental domain for the action of the \emph{extended modular group} $\text{PGL}_2(\Z)$. This action includes orientation reversing transformations so that adjacent triangles have opposite orientations. 
\end{remark}
\begin{figure}[h]
  \centering
  \includesvg[.8\textwidth]{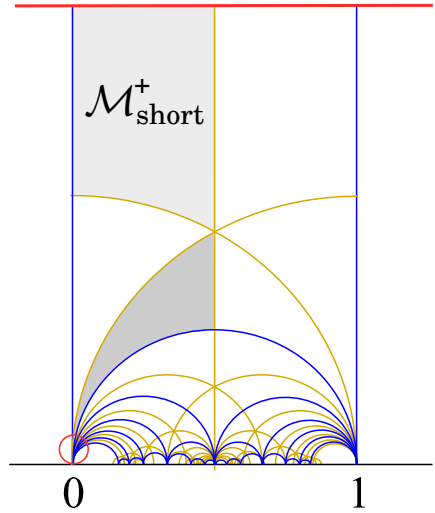}
  \caption{The image (in dark grey) of ${\cal M}_{\text{short}}^+$ by the action of  $\begin{pmatrix} 0 & 1 \\ -1 & 1\end{pmatrix}\in \text{SL}_2(\Z)$. The top horizontal line (in red) represents the horocycle $\tau_i=\Lmin$ (not to scale). Its image by $\begin{pmatrix} 0 & 1 \\ -1 & 1\end{pmatrix}$ is a circle (in red, not to scale) tangent  at 0 to the real axis.}
    \label{fig:ModularGroup-FundamentalDomain}
  \end{figure}
  \begin{lemma}\label{lem:short-tori}
    Any modulus in ${\cal M}_{\text{short}}^+$ can be geometrically realized by a diplotorus with parameters $n=19$ and $d\in\{2,7,13\}$.
  \end{lemma}
  \begin{proof}
  We need to check that ${\cal M}_{\text{short}}^+$ is covered by the orbit of ${\cal M}_{19,2}\bigcup {\cal M}_{19,7}\bigcup {\cal M}_{19,13}$ under the action of $\text{SL}_2(\Z)$. Equivalently, writing $g\cdot {\cal M}_{\text{short}}^+$ for the image of ${\cal M}_{\text{short}}^+$ by $g\in \text{SL}_2(\Z)$, we must have that
  \[\bigcup_{g\in \text{SL}_2(\Z)}g^{-1}\cdot \big((g\cdot{\cal M}_{\text{short}}^+)\bigcap ({\cal M}_{19,2}\bigcup {\cal M}_{19,7}\bigcup {\cal M}_{19,13})\big)
  \]
  covers ${\cal M}_{\text{short}}^+$ (or any of its images). The regions ${\cal M}_{19,2}$ ${\cal M}_{19,7}$  and ${\cal M}_{19,13}$, deduced from the formulas in Lemma~\ref{lem:Arnouxetal}, are plotted in Figures~\ref{fig:domain-1} and~\ref{fig:domain-2}.

  Denote by ${\cal M}^+$ the part of the fundamental domain in Figure~\ref{fig:moduli-space} with non-negative real part. ${\cal M}^+$ is bounded by the geodesic triangle with vertices $0,e^{i\pi/3},\infty$. In particular, it contains ${\cal M}_{\text{short}}^+$.
  Consider the subset $\Delta$ of $\text{SL}_2(\Z)$ composed of the matrices  $g_\delta:=\begin{pmatrix} 0 & 1 \\ -1 & \delta\end{pmatrix}$ with $\delta$ a positive integer. The elements in $\Delta$ transform ${\cal M}^+$ into a fan of triangular domains with positive real part, tangent at 0 to the imaginary axis\footnote{The orbit of ${\cal M}^+$ under the action of $\text{SL}_2(\Z)$ only includes the positively oriented triangles of the Dedekind tessellation. In particular, the  orbit of ${\cal M}^+$ by $\Delta$ includes one out of every two triangles in the fan.}.
  They moreover transform each horocycle $\{\tau_i= \text{constant}\}$ into a same circle tangent at 0 to the real axis.  Larger constants give rise to smaller circles.
   \begin{figure}[h]
    \centering
    \includesvg[\textwidth]{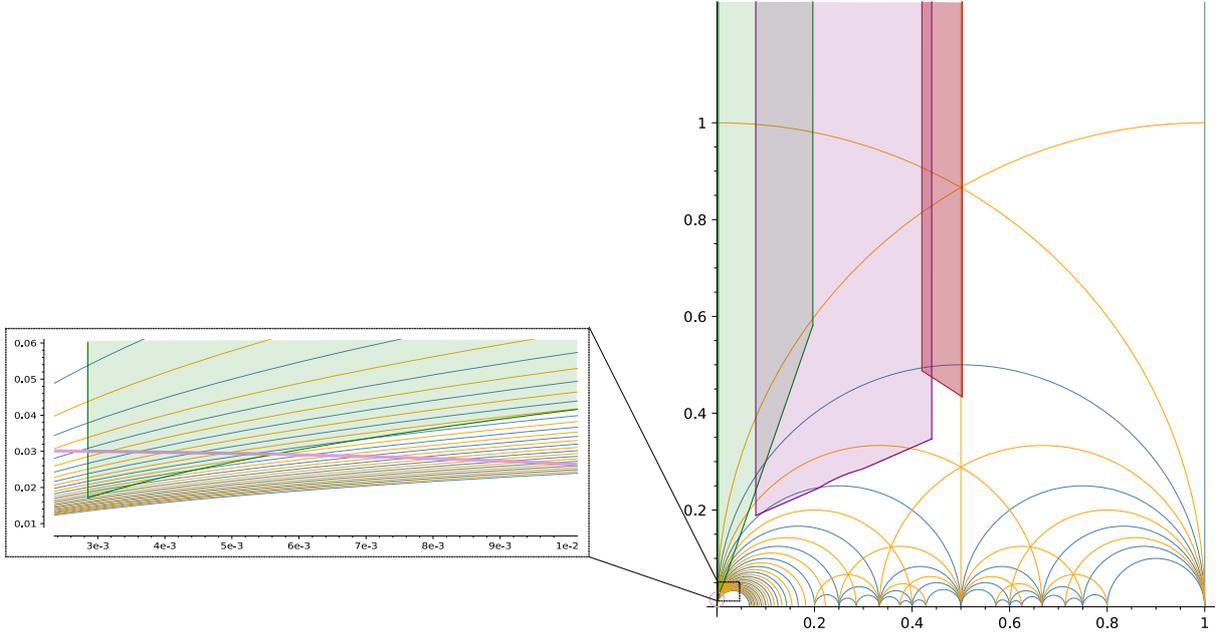}
    \caption{The moduli space for ${\cal M}_{19,2}$ (in green), ${\cal M}_{19,7}$ (in purple), and ${\cal M}_{19,13}$ (in red). The axes in the bottom enlargement have different scales; the horocycle (thick purple line), image of $\{\tau_i=\Lmin\}$ by the elements of $\Delta$, appear as an arc of ellipse.}
      \label{fig:domain-1}
  \end{figure}
  Two such circles cut the transforms of ${\cal M}^+$ by $\Delta$ into slices that are themselves transforms of a \emph{same} slice in ${\cal M}^+$. Figures~\ref{fig:slicing} and~\ref{fig:domain-2} demonstrates that we can slice 
  ${\cal M}_{\text{short}}^+$ so that each slice has a transform by respectively $g_1,g_3,g_5$ (in yellow, blue, and red on the figure) covered by ${\cal M}_{19,2}\bigcup {\cal M}_{19,7}\bigcup {\cal M}_{19,13}$.
  \begin{figure}[h]
    \centering
    \includesvg[.5\textwidth]{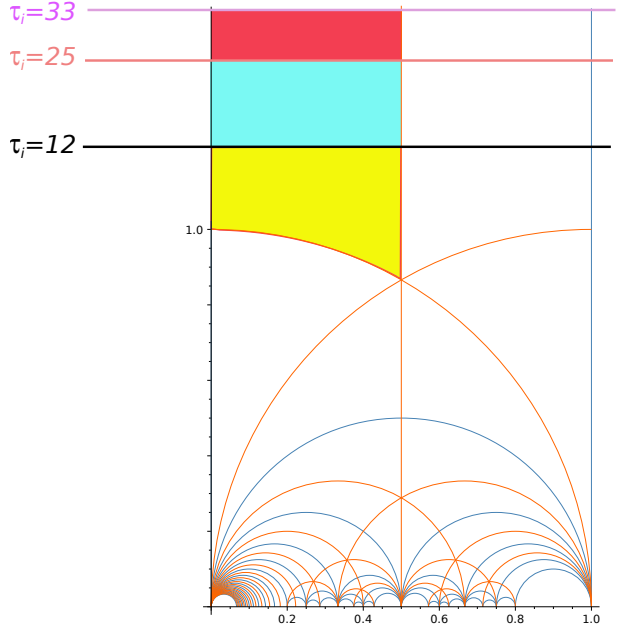}
    \caption{A slicing of ${\cal M}_{\text{short}}^+$ (not to scale) in three regions (yellow, blue and red) bounded by the horocycles $\tau_i=\Lmin, 25, 12$ (respectively in purple, orange, and black) and the Euclidean unit circle.}
    \label{fig:slicing}
  \end{figure}
     \begin{figure}[h]
    \centering
    \includegraphics[width=\textwidth]{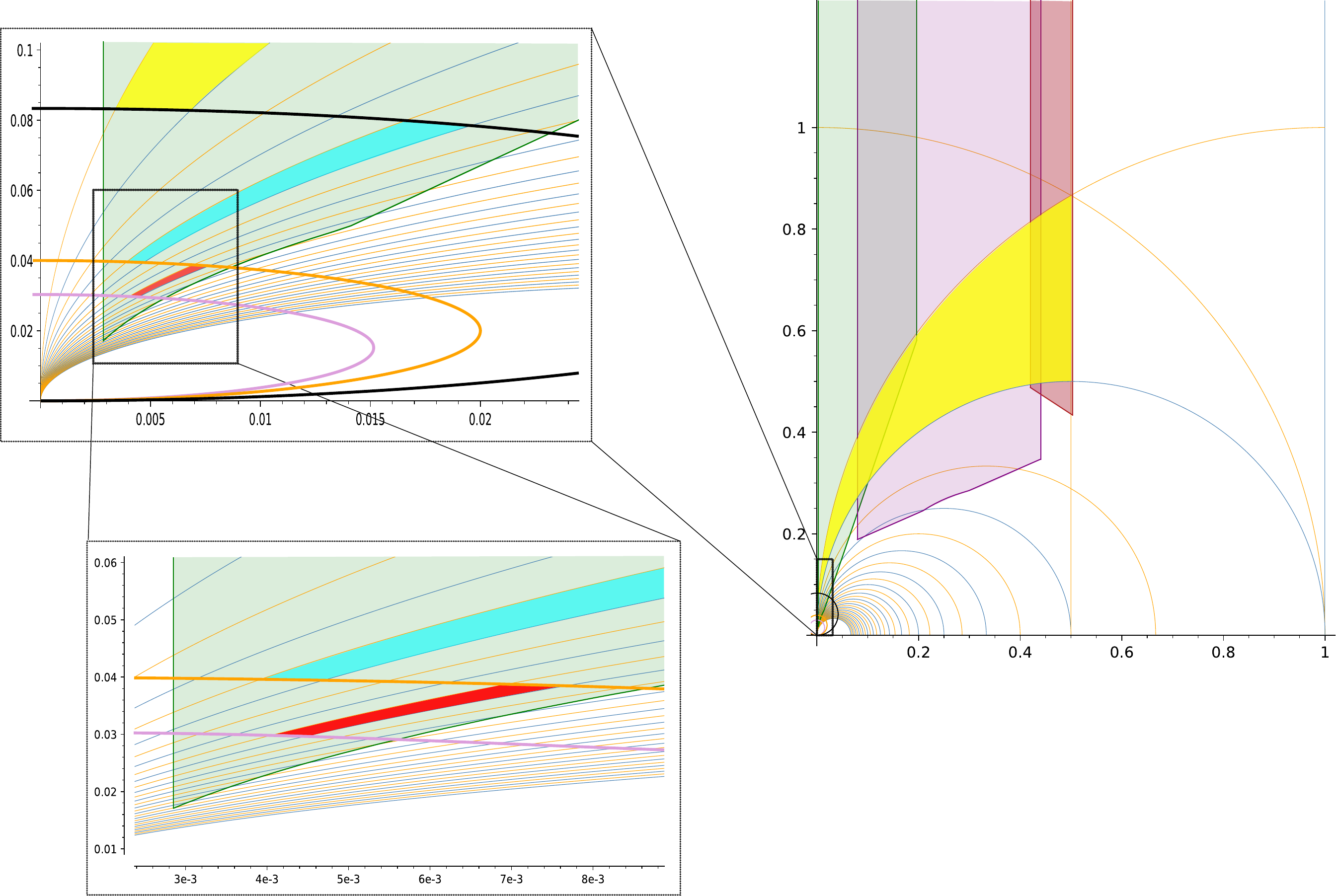}
    \caption{The three regions in red, blue and yellow are images of the corresponding slices in ${\cal M}_{\text{short}}^+$ by $g_1,g_3,g_5$, respectively.}
    \label{fig:domain-2}
  \end{figure}
  It follows that
  \[{\cal M}_{\text{short}}^+\subset \bigcup_{\delta\in\{1,3,5\}}g_\delta^{-1}\cdot \big((g_\delta\cdot{\cal M}_{\text{short}}^+)\bigcap ({\cal M}_{19,2}\bigcup {\cal M}_{19,7}\bigcup {\cal M}_{19,13})\big),
  \]
  which proves the lemma.
  This proof by picture can be made formal by computing the exact arrangement of the involved domains, whose boundaries are made of arcs of circles and line segments. The details are provided in Appendix.
  \end{proof}

  From Lemma~\ref{lem:short-tori} we can construct a universal triangulation for short tori. Indeed, all the diplotori with fixed parameters $n,d$ have the same abstract triangulation, that we denote by ${\cal T}_{n,d}$. Hence, we  just need a common subdivision of ${\cal T}_{19,2}$, ${\cal T}_{19,7}$ and ${\cal T}_{19,13}$
  to obtain such a universal triangulation. In the conference version of this article~\cite{lt-utft-22}, we computed a common subdivision with 2204 triangles. Here, we reduce this size by replacing\footnote{We thank the anonymous reviewers for suggesting this improvement.}  ${\cal T}_{19,13}$ by ${\cal T}_{19,6}$. Indeed, considering the symmetric of a diplotorus with respect to a horizontal plane, we easily deduce the existence of an orientation reversing isomorphism between ${\cal T}_{n,d}$ and ${\cal T}_{n,n-d}$. It is thus enough to compute a common refinement of ${\cal T}_{19,2}$, ${\cal T}_{19,7}$ and ${\cal T}_{19,6}$ to obtain a universal triangulation for short tori.
These three triangulations are obtained by identifying the boundaries of a same triangulated cylinder. However, they are not isomorphic, as one needs to apply distinct circular shifts before identification. We can nonetheless send them in a \emph{same} torus as follows. For $k\in\Z$, consider  the points
  \[A_k =(k,-1), \quad B_k=(k,0), \quad C_k=(k,1)
  \]
  in the infinite plane strip ${\cal B}:= \R\times [-1,1]$. Then, ${\cal T}_{19,d}$ is isomorphic to the triangulation of $\cal B$ by the triangles \\
  $\{A_kA_{k+1}B_k, B_kA_{k+1}B_{k+1}, C_kC_{k+1}B_{k-d}, B_{k-d}C_{k+1}B_{k+1-d}\}_{k\in\Z}$ quotiented by the horizontal translations generated by the vector $(19,0)$, further identifying the two boundaries according to the vertical translation $(0,2)$. This quotient and identification being independent of $d$, the three triangulations for $d=2,6,7$ are indeed embedded in a same torus; see Figure~\ref{fig:overlay-short-2-6-7}.
  \begin{figure}[h]
    \centering
    \includegraphics[width=\textwidth]{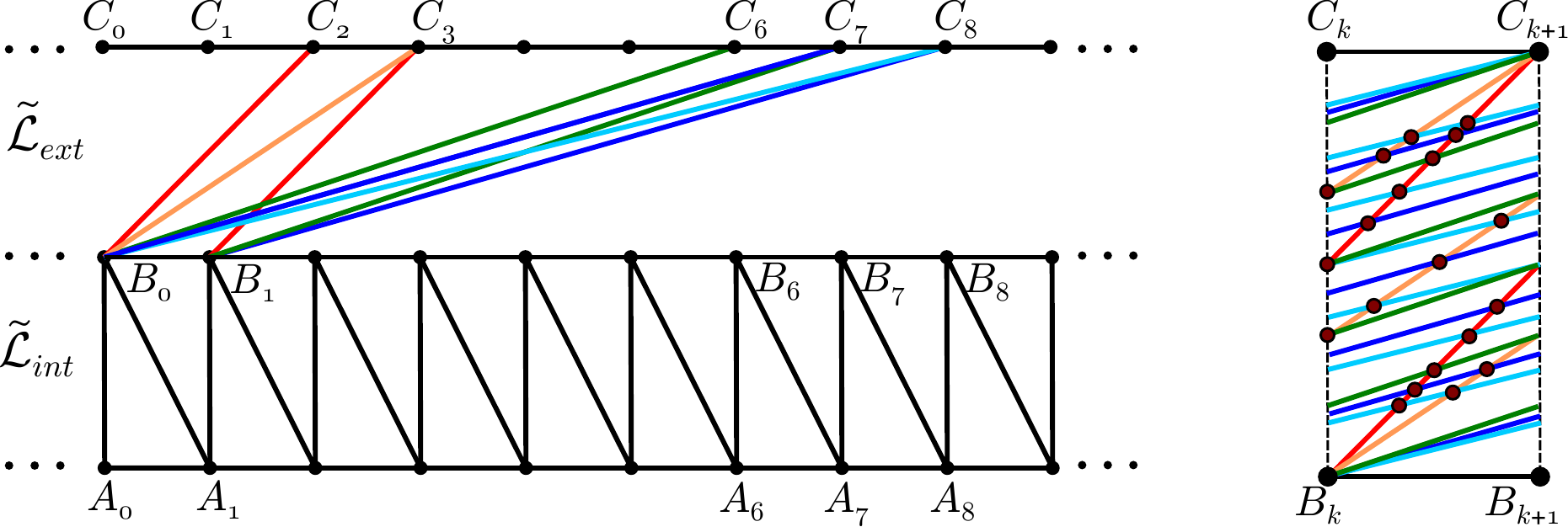} 
    \caption{Layout of the triangulations ${\cal T}_{19,2}$, ${\cal T}_{19,6}$ and ${\cal T}_{19,7}$. Left, the two sub-strips $\tilde{{\cal L}}_{\text{int}}$ and $\tilde{{\cal L}}_{\text{ext}}$ correspond to the (lift of) overlay of the internal and external ploids. Right, a period of $\tilde{{\cal L}}_{\text{ext}}$ contains 20 intersection points from the overlays of the external ploids.}
    \label{fig:overlay-short-2-6-7}
  \end{figure}

  We overlay the three triangulated strips obtained for $d=2,6,7$. We want to count the number of vertices of the resulting subdivision. We only have to care about the edges $B_kC_\ell$, the other ones being common to the three triangulations. The strip ${\cal B}$ being 1-periodical in the horizontal direction, it suffices to consider the number of intersections with the other edges of the 5 edges $B_0C_k$ for $k\in I:=\{2,3,6,7,8\}$. The edges $B_0C_k$ and $B_\ell C_{\ell+j}$, $j\in I$, intersect in their interior if and only if
  \begin{itemize}
  \item[$\bullet$] $\ell<0$ and $\ell+j>k$, or equivalently $k-j<\ell<0$, or
    \item[$\bullet$]  $\ell>0$ and $\ell+j<k$, or equivalently $0<\ell<k-j$.
  \end{itemize}
  In this case, we compute $p_{\ell,j}^k:=B_0C_k\bigcap B_\ell C_{\ell+j} = \frac{\ell}{k-j}(k,1)$. All other intersection points are horizontal translates of the $p_{\ell,j}^k$ by an integral amount. The set of intersection points with $x$-coordinate in $[0,1)$ is thus given by $\{(\text{frac}(\frac{\ell k}{k-j}), \frac{\ell}{k-j})\}_{j,k,\ell}$, where $j,k,\ell$ vary as above and $\text{frac}(x)$ is the fractional part of $x$. Eliminating the duplicates, we found 20 intersections leading to $n\times 20=380$ intersection points in total. See Figure~\ref{fig:overlay-short-2-6-7}.
  Adding the remaining points $A_k, B_k$ ($C_k$ and $A_k$ should be identified) we find a total of $380+ 38 = 418$ vertices. It remains to triangulate the subdivision by adding diagonals in the non triangular faces. By Euler's formula on the torus, we conclude that the triangulated overlay has \TshortT{} triangles. We have thus proved
  \begin{proposition}\label{prop:short-tori}
    There exists an abstract triangulation with  \TshortT{} triangles, which admits
linear embeddings isometric to every short torus.
\end{proposition}

\subsection{Overlaying short and long tori}\label{subsec:merging-short-long} 
It remains to overlay our universal triangulations for long and short tori to obtain a universal triangulation for all tori. Before overlaying the layouts  of Figures~\ref{fig:patronLongTori-ter} and~\ref{fig:overlay-short-2-6-7}, we perform some modifications. We first remove the diagonals introduced to triangulate the rectangular faces of the bends as they are not necessary to define the isometric PL embeddings of long tori. For the same reason, we remove the diagonals used to triangulate the three portions of right prisms; Compare Figure~\ref{fig:patronLongTori-ter} and top Figure~\ref{fig:shortLongOverlay}. 
\begin{figure}[h]
  \centering
  \includegraphics[width=\textwidth]{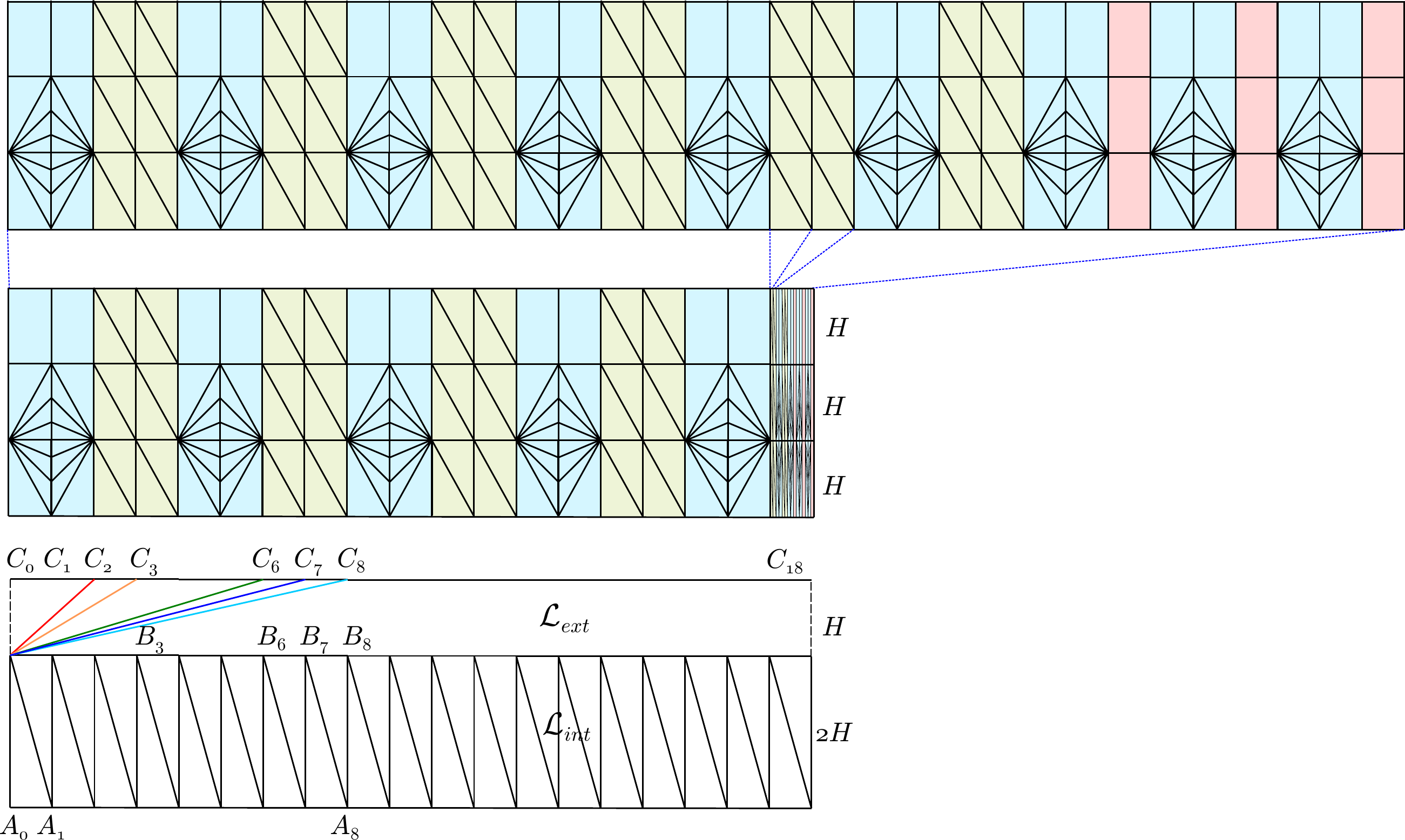}
  \caption{Modified layout of the universal triangulations for long and short tori.}
  \label{fig:shortLongOverlay}
\end{figure}
Denote by ${\cal L}_{\text{long}}$ the resulting layout. It is composed of three horizontal strips, where the top one has no internal vertex. ${\cal L}_{\text{long}}$  also divides into 33 vertical bands corresponding to 9 bends (each made of two bands), 12 gaskets, and 3 portions of right prisms.
We squeeze the last 15 bands to the width of a single one. We now view ${\cal L}_{\text{long}}$ as composed of 19 bands of equal length, where the last one contains the 15 squeezed bands. See middle Figure~\ref{fig:shortLongOverlay}. Denote by $H$ the common height of the three horizontal strips of ${\cal L}_{\text{long}}$.
We next consider the layout for short tori, call it ${\cal L}_{\text{short}}$. It decomposes into two horizontal strips ${\cal L}_{\text{int}}\bigcup {\cal L}_{\text{ext}}$ corresponding to the internal and external ploids; compare Figure~\ref{fig:overlay-short-2-6-7} with bottom of Figure~\ref{fig:shortLongOverlay}.
The bottom strip ${\cal L}_{\text{int}}$ divides into 19 vertical bands that we align with the ones of ${\cal L}_{\text{long}}$. We also stretch ${\cal L}_{\text{int}}$ and ${\cal L}_{\text{ext}}$ vertically so that their heights becomes respectively $2H$ and $H$.
We are now ready to overlay ${\cal L}_{\text{short}}$ and ${\cal L}_{\text{long}}$ as shown at the bottom of Figure~\ref{fig:shortLongOverlay}. Note that the (universal) subdivisions for short and long tori are obtained from the corresponding layouts by applying the \emph{same} identifications of their horizontal and vertical sides. Applying these identifications to the overlay of the layouts thus provides a common refinement of the subdivisions for short and long tori.

To enumerate the vertices of the overlay we consider adding the edges of ${\cal L}_{\text{short}}$ to the layout ${\cal L}_{\text{long}}$. The horizontal and vertical edges of ${\cal L}_{\text{short}}$ can be mapped to the corresponding edges of ${\cal L}_{\text{long}}$ without introducing new vertices. Note that the three horizontal edges in the last band of ${\cal L}_{\text{short}}$ are each subdivided into 15 edges. The diagonals of ${\cal L}_{\text{int}}$ are inserted as follows. Among the 18 first diagonals, the ones inserted in a band corresponding to a bend  in ${\cal L}_{\text{long}}$ introduce 6 crossings each, while the ones inserted in a band corresponding to a gasket  introduce 1 crossing each. See bottom left of Figure~\ref{fig:shortLongOverlay-bis}.
\begin{figure}[h]
  \centering
  \includegraphics[width=\textwidth]{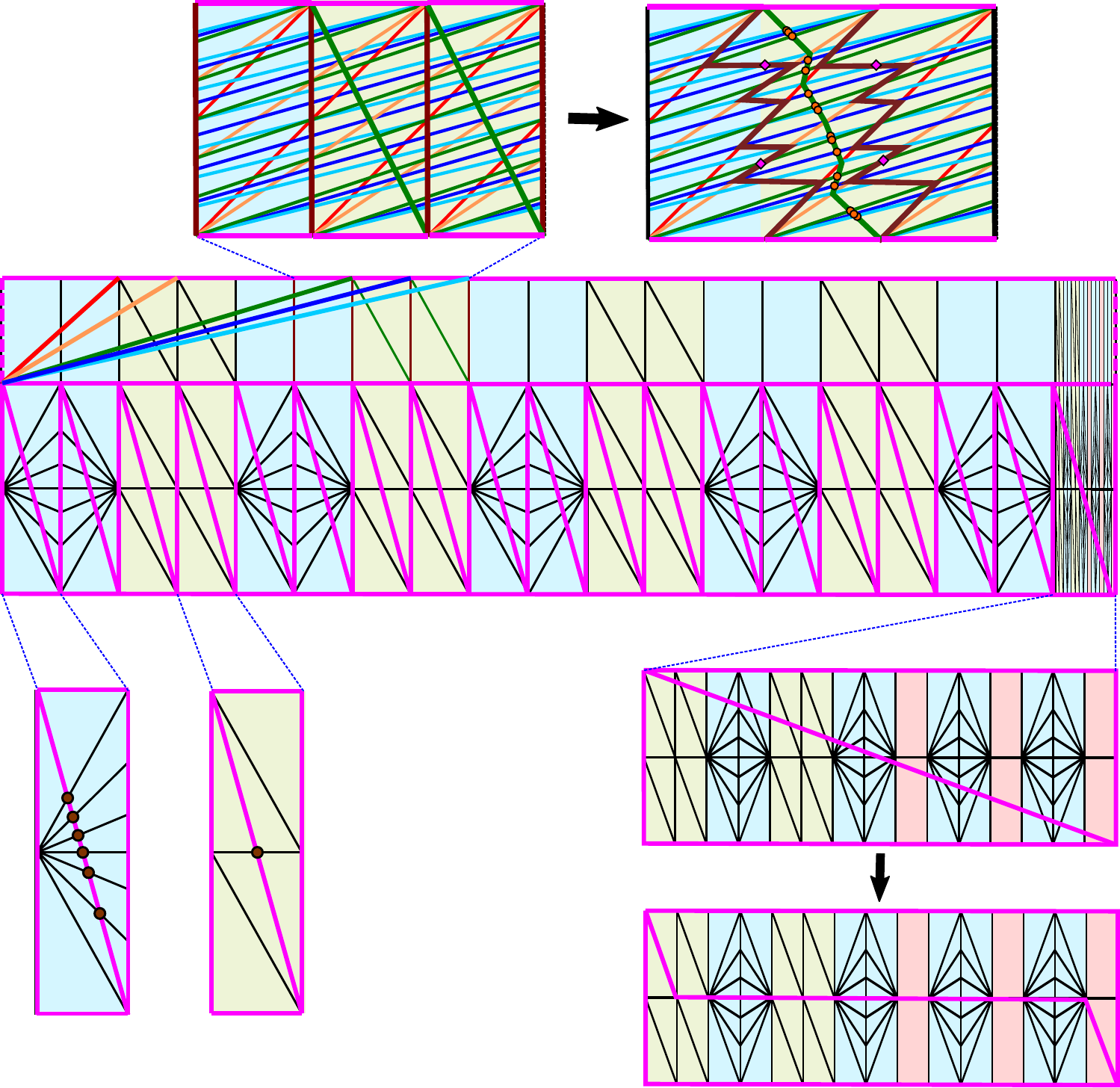}
  \caption{The subdivision of the remaining two diagonals of ${\cal L}_{\text{int}}$ (thick purple lines).}
  \label{fig:shortLongOverlay-bis}
\end{figure}
This introduces $6E_b+E_g$ vertices in total, where $E_b=10$ and $E_g=8$ are the respective numbers of bands of type bend  and gasket. After subdivision, the 19th diagonal can be inserted without introducing any vertex as shown at the bottom right of Figure~\ref{fig:shortLongOverlay-bis}. Apart from their own intersections, the slanted edges of ${\cal L}_{\text{int}}$ cross the vertical and diagonal edges in the upper strip of ${\cal L}_{\text{long}}$.  We deform these edges in the 18 first bands in order to follow as much as possible the subdivision of ${\cal L}_{\text{ext}}$ so as to minimize the number of new crossings. See the top part of Figure~\ref{fig:shortLongOverlay-bis}. This way, the slanted edges of ${\cal L}_{\text{int}}$  introduce 2 crossings per vertical edge and 16 crossings per diagonal of ${\cal L}_{\text{long}}$. In total, this leads to $2E_v + 16E_d$ crossings, where $E_v=19$ and $E_d=8$ stand for the respective numbers of verticals and diagonals in the 18 first bands of the upper strip. Finally, each vertical or diagonal in the interior of the squeezed part of ${\cal L}_{\text{long}}$ crosses all the slanted edges $B_\ell C_{\ell+j}$ of ${\cal L}_{\text{int}}$ that surround its extremities. Since the considered verticals and diagonals are squeezed in the interior of a band of width 1, this leads to $\sum_{j\in I}j=26$ crossings each. In total we thus get $26(E'_v + E'_d)$ crossings in the squeezed part, where $E'_v=14$ and $E'_d=4$ denote the respective numbers of vertical and diagonal edges.
To sum up, the overlay contains
\begin{itemize}
\item $V_{\text{long}}=\TlongT/2=\VlongT$ vertices from ${\cal L}_{\text{long}}$ as counted in Proposition~\ref{th:univ-triang},
\item $V_\cap=380$ vertices from the intersecting edges in ${\cal L}_{\text{short}}$ as computed in the proof of Proposition~\ref{prop:short-tori},
\item $6E_b+E_g = 68$ vertices from the overlay of diagonals of ${\cal L}_{\text{int}}$ with ${\cal L}_{\text{long}}$,
\item $2E_v + 16E_d = 166$ vertices from the overlay of the slanted edges of ${\cal L}_{\text{ext}}$ with the 18 first bands of ${\cal L}_{\text{long}}$,
\item $26(E'_v + E'_d)=468$ vertices from the overlay of the slanted edges of ${\cal L}_{\text{ext}}$ with squeezed part of ${\cal L}_{\text{long}}$.
  \end{itemize}
In total the overlay thus contains
\[V_{\text{long}}+V_\cap + 6E_b+E_g + 2E_v + 16E_d + 26(E'_v + E'_d)=\VlongT +380+68  +166+468  =\optV
  \]
vertices. By Euler's formula this corresponds, after adding diagonals to triangulate the overlay, to \optT{} triangles.
This ends the proof of Theorem~\ref{th:univ-triang}.

\section{Conclusion and open problems}

While our construction provides a triangulation of reasonable size, it can certainly be optimized. First, the size of our universal triangulation for long tori can probably be reduced, say by simplifying the helical twist. Then, the choice of the diplotori to realize the short tori may also be optimized to reduce the size of the resulting triangulation. Finally, the overlay of these triangulations can also be optimized. A challenging question is to find the smallest number of triangles in a universal triangulation for flat tori. In fact, an analogous question can be formulated for a single flat torus, asking for the minimum size of a geometric realization. This  already seems not obvious. What is, for instance, the least number of vertices in a geometric realization of the square flat torus?

Once a universal triangulation is chosen, is it possible to connect any two geometric realizations of this universal triangulation by a continuous
deformation through flat tori. Said differently, is it true that the deformation space of the universal triangulation restricted to flat tori is connected, or has a connected component covering the whole modular curve?

Another intriguing challenge is to prove the existence of a universal triangulation for other moduli spaces. Good candidates can be found in the strata of the so-called translation surfaces. However, the set of points at infinity, reduced to a single point for the moduli space of tori, becomes much more complicated to handle already in the case of genus two surfaces.

 \section*{Acknowledgements}
 We warmly thank Alba M\'alaga, Pierre Arnoux and Samuel Leli\`evre for sharing with us their constructions of flat tori and showing us how to cover their moduli space with these constructions. Figures~\ref{fig:domain-1} and~\ref{fig:domain-2} were computed thanks to their Sage program. We also thank the anonymous reviewers for their many suggestions, including the improvement for the size of the overlay of the universal triangulations of short and long tori.

\appendix
\section{Appendix}\label{sec:appendix} 
Here, we provide the details for the proof of Lemma~\ref{lem:short-tori}.
From Lemma~\ref{lem:Arnouxetal}, simple computations show that the region ${\cal M}_{19,2}$ is bounded by the following parametrized curves:
\begin{itemize}
\item{$\lambda_2(t) = z_2 + it$ with $t \in [0, +\infty[$ and $z_2 = \frac{2 - \sin\frac{2\pi}{19} \cot\frac{\pi}{19} + i \sin\frac{2 \pi}{19}}{19}$,}
\item{$\beta_{2,1}(t) = \frac{2}{19} - \frac{\sin\frac{2\pi}{19}}{19 \sin\frac{\pi}{19}} e^{-it}$ with $t \in [\frac{\pi}{19}, \frac{3 \pi}{19}]$,}
\item{$\beta_{2,2}(t) = \frac{2 + i\cot\frac{\pi}{19}}{19} - \frac{e^{i\frac{15\pi}{38}}}{19 \sin\frac{\pi}{19}} t$ with $t \in [\cos\frac{16\pi}{19}, \cos\frac{3\pi}{19}]$,}
\item{$\rho_2(t) = w_2 + it$ with $t \in [0, +\infty[$ and $w_2 = \beta_{2,2}\left(\cos \frac{16 \pi}{19}\right)$.}
\end{itemize}
While $\mathcal{M}_{19,7}$ is bounded by:
\begin{itemize}
\item{$\lambda_7(t) = z_7 + it$ with $t \in [0, +\infty[$ and $z_7 = \frac{7 - \sin\frac{7\pi}{19} \cot\frac{\pi}{19} + i \cot\frac{\pi}{19} \left(1 - \cos\frac{7\pi}{19}\right)}{19}$,}
\item{$\beta_{7,1}(t) = \frac{7 + i \cot\frac{\pi}{19}}{19} - \frac{e^{i \frac{5\pi}{38}}}{19 \sin\frac{\pi}{19}} t$ with $t \in [\cos\frac{6\pi}{19}, \cos\frac{\pi}{19}]$,}
\item{$\beta_{7,2}(t) = \frac{7}{19} - \frac{\sin\frac{7\pi}{19}}{19 \sin\frac{\pi}{19}} e^{-it}$ with $t \in [\frac{6\pi}{19}, \frac{8\pi}{19}]$,}
\item{$\beta_{7,3}(t) = \frac{7 + i\cot\frac{\pi}{19}}{19} - \frac{e^{i \frac{5\pi}{38}}}{19 \sin\frac{\pi}{19}} t$ with $t \in [\cos\frac{11\pi}{19}, \cos\frac{8\pi}{19}]$,}
\item{$\rho_7(t) = w_7 + it$ with $t \in [0, +\infty[$ and $w_7 = \beta_{7,3}\left(\cos\frac{11\pi}{9}\right)$.}
\end{itemize}
And $\mathcal{M}_{19,13}$ is bounded by:
\begin{itemize}
\item{$\lambda_{13}(t) = z_{13} + it$ with $t \in [0, +\infty[$ and $z_{13} = \frac{13 - \sin\frac{13\pi}{19} \cot\frac{\pi}{19} + i \cot\frac{\pi}{19} \left(1 - \cos\frac{13\pi}{19}\right)}{19}$,}
\item{$\beta_{13,1}(t) = \frac{13 + i \cot\frac{\pi}{19}}{19} - \frac{e^{-i\frac{7\pi}{38}}}{19 \sin\frac{\pi}{19}} t$ with $t \in [\cos\frac{5 \pi}{19}, \cos\frac{\pi}{19}]$,}
\item{$\rho_{13}(t) = w_{13} + it$ with $t \in [0, +\infty[$ and $w_{13} = \beta_{13,1}\left(\cos\frac{5\pi}{19}\right)$.}
\end{itemize}
In the sequel we denote by $\mathcal{B}(z, r)$ the closed disk of radius $r$ centered at $z$ and by $\mathcal{C}(z, r)$ its boundary circle. We also denote by $\Re(z)$ and $\Im(z)$, the real part and the imaginary part of $z$, respectively.
Recall that each $g_{\delta} : z \mapsto \frac{1}{-z+\delta}$ sends the horizontal line $\{\Im(z) = h\}$ onto the circle $\mathcal{C}(\frac{i}{2h}, \frac{1}{2h})$. Furthermore, for $\delta \geq 1$, $g_{\delta}$ sends the imaginary axis onto the circle $\mathcal{C}(\frac{1}{2\delta}, \frac{1}{2\delta})$, and the line $\{\Re(z) = \frac{1}{2}\}$ onto the circle $\mathcal{C}(\frac{1}{2\delta - 1}, \frac{1}{2\delta - 1})$. We denote the red, blue and yellow slices of ${\cal M}^+_\text{short}$ on Figure~\ref{fig:slicing} by respectively ${\cal S}_r$, ${\cal S}_b$ and ${\cal S}_y$.

From these facts, we deduce that $g_5({\cal S}_r)$, the image of the red slice by $g_5$ (see Figure~\ref{fig:domain-2}), is bounded by four arcs of circles; one from respectively $\mathcal{C}(\frac{1}{9}, \frac{1}{9}), \mathcal{C}(\frac{i}{66}, \frac{1}{66}), \mathcal{C}(\frac{1}{10}, \frac{1}{10})$ and $\mathcal{C}(\frac{i}{50}, \frac{1}{50})$. Similarly, the image $g_3({\cal S}_b)$ of the blue slice is bounded by arcs from the circles $\mathcal{C}(\frac{1}{5}, \frac{1}{5}), \mathcal{C}(\frac{i}{50}, \frac{1}{50}), \mathcal{C}(\frac{1}{6}, \frac{1}{6})$ and $\mathcal{C}(\frac{i}{24}, \frac{1}{24})$. Finally, the image $g_1({\cal S}_y)$ of the yellow slice is bounded by arcs of the circles $\mathcal{C}(1, 1), \mathcal{C}(\frac{i}{24}, \frac{1}{24}), \mathcal{C}(\frac{1}{2}, \frac{1}{2})$ and a segment of the line $\{ \Re(z) = \frac{1}{2} \}$. For this last slice, we note that $g_1$ sends the arc of circle $\{ e^{it} \mid \frac{\pi}{3} \leq t \leq \frac{\pi}{2}\}$ to the vertical line segment $\{ \frac{1}{2} + i t \mid t \in [\frac{1}{2}, \frac{\sqrt{3}}{2}] \}$.

We now proceed to prove that the three curvilinear quadrilaterals $g_5({\cal S}_r)$,  $g_3({\cal S}_b)$ and $g_1({\cal S}_y)$ shown in Figure~\ref{fig:domain-2} lie above the lower boundary of $\mathcal{M}_{19} := \mathcal{M}_{19,2} \cup \mathcal{M}_{19, 7} \cup \mathcal{M}_{19, 13}$. (A point $z$ lies above another point $w$ with same real part if $\Im(z) \geq \Im(w)$.) Let us remark that for showing that such a quadrilateral lies above some boundary, it suffices to show that the bottom side and the right most side of the quadrilateral lie above this boundary as the quadrilateral is completely included in the region of the plane above these two sides.

\paragraph{The Red quadrilateral $g_5({\cal S}_r)$.} Denote the vertices of this quadrilateral as in Figure~\ref{fig:quadrilateral}. From the above description, one computes $A_1 = \frac{6 + 44 i}{1479}, A_2 = \frac{5 + 33 i}{1114}, A_3 = \frac{1 + 5 i}{130}, A_4 = \frac{18 + 100 i}{2581}$. By the previous remark it is enough to show that the curvilinear sides $\overset{\frown}{A_1A_2}$ and $\overset{\frown}{A_2A_3}$ lie above the lower boundary of $\mathcal{M}_{19}$.

\begin{figure}
\centering
\includesvg[\linewidth]{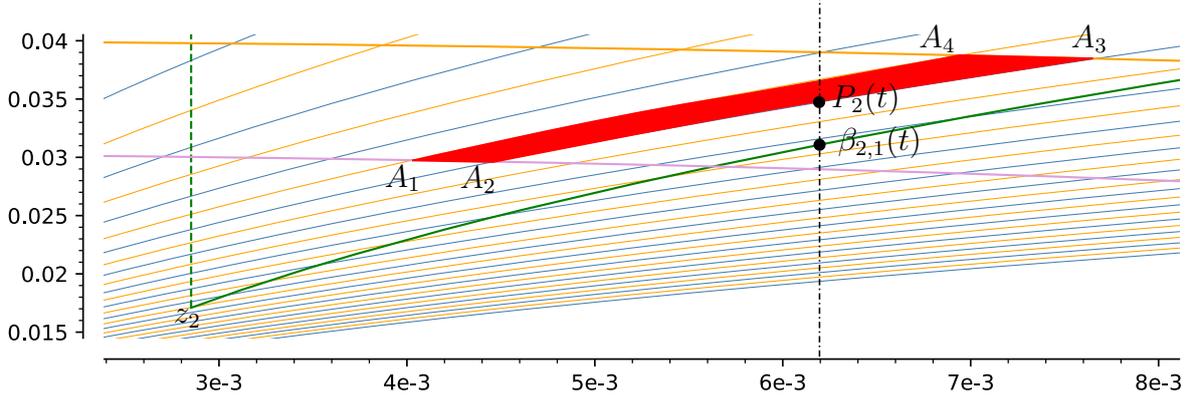}
\caption{The red quadrilateral $A_1A_2A_3A_4$}
\label{fig:quadrilateral}
\end{figure}

\begin{itemize}
\item Arc $\overset{\frown}{A_1A_2}$.

 The vertical line through $\beta_{2,1}(t)$ cuts $\mathcal{C}(\frac{i}{66}, \frac{1}{66})$ (the circle containing $\overset{\frown}{A_1A_2}$) in two points. We denote by $P_1(t)$ the highest of these two points. Remark that $\overset{\frown}{A_1A_2}$ is indeed in the upper half part of $\mathcal{C}(\frac{i}{66}, \frac{1}{66})$. We compute $P_1(t) = \frac{i + e^{i \tau}}{66}$ with $0 \leq \tau \leq \pi$, where $\frac{\cos\tau }{66}= \frac{2}{19} - \frac{\sin\frac{2\pi}{19}}{19 \sin\frac{\pi}{19}} \cos t$. Thus:
\begin{equation*}
\begin{aligned}
&\Im(P_1(t)) = \frac{1 + \sin\tau}{66} \geq \Im(\beta_{2,1}(t)) = \frac{\sin\frac{2\pi}{19}}{19\sin\frac{\pi}{19}} \sin t \\ 
&\iff \frac{\sin^2\tau }{66^2} \geq \left(-\frac{1}{66} + \frac{\sin\frac{2\pi}{19}}{19 \sin\frac{\pi}{19}}\sin t \right)^2 \\
&\iff f_1(\cos t)\geq 0,
\end{aligned}
\end{equation*}
where, 
\[f_1(x):= \frac{8x}{361}\cos\frac{\pi}{19} + \frac{2\sqrt{1 - x^2}}{627} \cos\frac{\pi}{19}  - \frac{6 + 2 \cos\frac{2\pi}{19}}{361} 
\]
A study of $f_1$ shows that it is non negative on $[0.96, 1]$, so that $f_1(\cos t)$ is non negative for $t \in [0, t_{red}^1 := 0.283]$. Note that $0<\pi/19<t_{red}^1<3\pi/19$. Moreover\footnote{Here and in the sequel, we write $x\approx y$, where $y=\sum_{i=k}^{\ell}d_i 10^i$ is a decimal with $d_k\neq 0$, to mean that $|x-y| < 10^{k}$.}, $\Re(z_2=\beta_{2,1}(\frac{\pi}{19})) \approx 0.002 < \Re(A_1) = \frac{6}{1479}$ and $\Re(\beta_{2,1}(t_{red}^1)) \approx 0.0055 > \Re(A_2) = \frac{5}{1114}$. This shows that the arc $\overset{\frown}{A_1A_2}$ entirely lies above $\beta_{2,1}$.

\item Arc $\overset{\frown}{A_2A_3}$.

  The vertical line through $\beta_{2,1}(t)$ cuts $\mathcal{C}(\frac{1}{10}, \frac{1}{10})$ (the circle containing $\overset{\frown}{A_2A_3}$) in two points. Denoting by $P_2(t)$ the highest of these two points, similar computations as above lead to:
  \[\Im(P_2(t)) \geq \Im(\beta_{2,1}(t)) \iff f_2(\cos t) :=  \frac{1}{1805}\frac{\sin\frac{2\pi}{19}}{\sin\frac{\pi}{19}} \cos t + \frac{18}{1805} - \frac{\sin^2\frac{2\pi}{19}}{361 \sin^2\frac{\pi}{19}} \geq 0.
  \]
  Since $f_2$ is non negative on $[0.55, 1]$, we have that $f_2(\cos t)$ is non negative for $t \in [0, 0.9]$. Note that this interval contains the interval of definition $[\frac{\pi}{19}, \frac{3 \pi}{19}]$ of $\beta_{2,1}$. Since $\Re(z_2) < \Re(A_2)$ and $\Re(\beta_{2,1}(\frac{3\pi}{19})) \approx 0.014 > \Re(A_3) = \frac{1}{130}$, the arc $\overset{\frown}{A_2A_3}$ is included in the region above $\beta_{2,1}$. We thus conclude that  $A_1A_2A_3A_4$ lies entirely above $\beta_{2,1}$.
\end{itemize}

\paragraph{The blue quadrilateral $g_3({\cal S}_b)$.} Denote the vertices of this quadrilateral $B_1, B_2, B_3, B_4$ in analogy with what precedes ($B_1$ is the bottom left vertex of the quadrilateral, $B_2$ the bottom right one, $B_3$ the top right one and $B_4$ the top left one). We compute $B_1 = \frac{2 + 20 i}{505}, B_2 = \frac{3 + 25 i}{634}, B_3 = \frac{1 + 4 i}{51}, B_4 = \frac{10 + 48 i}{601}$. Again, it is enough to show that the curvilinear sides $\overset{\frown}{B_1B_2}$ and $\overset{\frown}{B_2B_3}$ lie above the lower boundary of $\mathcal{M}_{19}$.

\begin{itemize}
\item Arc $\overset{\frown}{B_1B_2}$.

    The vertical line passing by $\beta_{2,1}(t)$ cuts $\mathcal{C}(\frac{i}{50}, \frac{1}{50})$ (the circle containing  $\overset{\frown}{B_1B_2}$) in two points. Let $Q_1(t)$ denotes the highest of these two points. We have
    \begin{align*}
      \Im(Q_1(t)) &\geq \Im(\beta_{2,1}(t)) \iff \\
      f_3(\cos t) &:= \frac{4 \sin\frac{2\pi}{19}}{361 \sin\frac{\pi}{19}} \cos t + \frac{\sin\frac{2\pi}{19}}{475 \sin\frac{\pi}{19}} \sqrt{1 - \cos^2 t} -\frac{4}{361} - \frac{\sin^2\frac{2\pi}{19}}{361 \sin^2\frac{\pi}{19}} \geq 0.
    \end{align*}

    Furthermore, $f_3$ is non negative on $[0.94, 1]$, so that $f_3(\cos t)$ is non negative on $[0, t_{blue}^1 := 0.3]$. Note that $0 < \pi/19 <0.3 < 3\pi/19$. Moreover, $\Re(z_2) \approx 0.002 < \Re(B_1) = \frac{2}{505}$ and $\Re(\beta_{2,1}(t_{blue}^1)) \approx 0.008 > \Re(B_2) = \frac{3}{634}$, which implies that $\overset{\frown}{B_1B_2}$ lies entirely above $\beta_{2,1}$.
\item Arc $\overset{\frown}{B_2B_3}$.

  Denote by $Q_2(t)$ the highest intersection point of the vertical line passing through $\beta_{2,1}(t)$ with $\mathcal{C}(\frac{1}{6}, \frac{1}{6})$ (the circle containing $\overset{\frown}{B_2B_3}$). We have
  \[\Im(Q_2(t)) \geq \Im(\beta_{2,1}(t)) \iff f_4(\cos t) := - \frac{7 \sin\frac{2\pi}{19}}{1083 \sin\frac{\pi}{19}} \cos t + \frac{26}{1083} - \frac{\sin^2\frac{2\pi}{19}}{361 \sin^2\frac{\pi}{19}} \geq 0.
  \]
  Since $f_4$ is non negative on $[-1, 1]$ and $\Re(z_2) < \Re(B_2)$, it follows that $\overset{\frown}{B_2B_3}$ is above $\beta_{2,1}$ over the interval $[\Re(B_2), \Re(\beta_{2,1}(\frac{3\pi}{19}))]$.
Let $\overline{\beta}_{2,2}$ be the supporting line of $\beta_{2,2}$. The point of $\overline{\beta}_{2,2}$ on the same vertical as $B_2$ is $Q_3 := \beta_{2,2}\left( \frac{1211 \sin\frac{\pi}{19}}{634 \cos\frac{15\pi}{38}} \right)$, while the point of $\beta_{2,2}$ on the same vertical line as $B_3$ is $Q_4 := \beta_{2,2}\left( t_{blue}^2 := \frac{83 \sin\frac{\pi}{19}}{51 \cos\frac{15\pi}{38}} \right)$. Observe that $t_{blue}^2 \in [\cos\frac{16\pi}{19}, \cos\frac{3\pi}{19}]$. We compute $\Im(Q_3) \approx 0.02 < \Im(B_2) = \frac{25}{634}$ and $\Im(Q_4) \approx 0.06 < \Im(B_3) = \frac{4}{51}$. By concavity of $\overset{\frown}{B_2B_3}$, we deduce that $\overset{\frown}{B_2B_3}$ lies above $\beta_{2,2}$ over $[\Re(\beta_{2,2}(\cos \frac{3\pi}{19})),\Re(B_3)]$. We conclude that $\overset{\frown}{B_2B_3}$ lies above $\beta_{2,1}\cup \beta_{2,2}$.
\end{itemize}

\paragraph{The yellow quadrilateral $g_1({\cal S}_y)$.} Denote the vertices of this quadrilateral $C_1, C_2, C_3, C_4$ analogously to what precedes.
One computes 
\[C_1 = \frac{2 + 48 i}{577}, C_2 = \frac{1 + 12 i}{145}, C_3 = \frac{1 + i}{2}, C_4 = \frac{1 + \sqrt{3} i}{2} = e^{i \frac{\pi}{3}}.
  \]
We show that the curvilinear sides $\overset{\frown}{C_1C_2}$ and $\overset{\frown}{C_2C_3}$ lie above the lower boundary of $\mathcal{M}_{19}$.
\begin{itemize}
\item Arc $\overset{\frown}{C_1C_2}$.

  Let $R_1(t)$ be the highest intersection point of the vertical through $\beta_{2,1}(t)$ with $\mathcal{C}(\frac{i}{24}, \frac{1}{24})$ (the circle containing $\overset{\frown}{C_1C_2}$). We have
  \begin{align*}
    \Im(R_1(t)) &\geq \Im(\beta_{2,1}(t)) \iff \\
    f_5(\cos t) &:= \frac{8 \cos\frac{\pi}{19}}{361} \cos t + \frac{\cos\frac{\pi}{19}}{114} \sqrt{1 - \cos^2 t} - \frac{6 + 2 \cos\frac{2\pi}{19}}{361} \geq 0.
  \end{align*}
  Since $f_5$ is non negative on $[0.8, 1]$, it ensues that $f_5(\cos t)$ is non negative for $t$ in $[0, 0.6]$. This interval contains $[\frac{\pi}{19}, \frac{3 \pi}{19}]$. As $\Re(z_2) \approx 0.002 < \Re(C_1) = \frac{2}{577}$ and $\Re(\beta_{2,1}(\frac{3\pi}{19})) \approx 0.02 > \Re(C_2) = \frac{1}{145}$, we conclude that $\overset{\frown}{C_1C_2}$ lies entirely above $\beta_{2,1}$.

\item Arc $\overset{\frown}{C_2C_3}$.
  
First we note that if a point $z$ with $\Re(C_3) \leq \Re(z) \leq \frac{1}{2}$ belongs to $\mathcal{B}(\frac{1}{2}, \frac{1}{2})$ then it lies below the arc $\overset{\frown}{C_2C_3}$. Thus to show that $\beta_{i, j}(t)$ is below $\overset{\frown}{C_2C_3}$ it is sufficient to show that $|\beta_{i, j}(t) - \frac{1}{2}| \leq \frac{1}{2}$.

We have
\[|\beta_{2,1}(t) - \frac{1}{2}| \leq \frac{1}{2} \iff f_6(\cos t) := \frac{15 \sin\frac{2\pi}{19}}{361 \sin\frac{\pi}{19}} \cos t + \frac{\sin^2\frac{2\pi}{19}}{361 \sin^2\frac{\pi}{19}} - \frac{34}{361} \leq 0.
\]
Since $f_6$ is non positive on $[-1, 1]$, it follows that $f_6(\cos t)$ is always non positive. Hence $\overset{\frown}{C_2C_3}$ is above  $\beta_{2,1}$ over $[\Re(C_2), \Re(\beta_{2,1}(\frac{3\pi}{19}))]$. 

Next, we show that $\overset{\frown}{C_2C_3}$ lies above $\beta_{2,2}$ in the strip $\{z\mid \Re(z)\in [\Re(\beta_{2,1}(\frac{3\pi}{19})), \Re(z_7)]\}$. Let $R_{2,7}$ be the point on $\beta_{2,2}$ with real part $\Re(z_7)$. We have\\ $R_{2,7} = \beta_{2,2}\left(\tau_{2, 7} := \frac{(-5 + \cos\frac{5\pi}{38} \cot\frac{\pi}{19})\sin\frac{\pi}{19}}{\sin\frac{2\pi}{19}}\right)$ and we verify that $\tau_{2, 7} \in [\cos \frac{16 \pi}{19}, \cos\frac{3 \pi}{19}]$ and $\Im R_{2,7} \approx 0.24 > \Im z_7 \approx 0.18$. Then, to show that $\overset{\frown}{C_2C_3}$ lies above $\beta_{2,2}$ in the above strip, it suffices by concavity of $\mathcal{B}(\frac{1}{2}, \frac{1}{2})$ (as $\beta_{2,2}$ is a line segment) to show that $\beta_{2,2}(\cos\frac{3\pi}{19}), R_{2,7} \in \mathcal{B}(\frac{1}{2}, \frac{1}{2})$. We indeed compute: $|\beta_{2,2}(\cos\frac{3\pi}{19}) - \frac{1}{2}|^2 \approx 0.23 < \frac{1}{4}$, $|R_{2,7} - \frac{1}{2}|^2 \approx 0.23 < \frac{1}{4}$.

It remains to show that the lower boundaries of $M_{19, 7}$ and $M_{19, 13}$ lie below $\overset{\frown}{C_2C_3}$ in the strip $\{z\mid \Re(z)\in [\Re(z_7), \Re(C_3)=1/2]\}$.

We have
\begin{align*}
  &|\beta_{7,1}(t) - \frac{1}{2}|^2 = |\beta_{7,3}(t) - \frac{1}{2}|^2 \leq \frac{1}{4} \iff \\
  &f_7(t) := \frac{\cos^2\frac{5\pi}{38}}{361 \sin^2\frac{\pi}{19}} t^2 + \frac{5 \cos\frac{5\pi}{38}}{361 \sin\frac{\pi}{19}} t - \frac{84}{361} +  \frac{\cos^2\frac{\pi}{19} + \cos\frac{\pi}{19} \sin\frac{5\pi}{38} + \sin^2\frac{5 \pi}{38}}{361 \sin^2\frac{\pi}{19}} \leq 0.
 \end{align*}
 As $f_7$ is non positive on the interval $[-1, 1]$, which contains the domains of $\beta_{7,1}$ and $\beta_{7,3}$, we deduce that these two curves lie entirely below $\overset{\frown}{C_2C_3}$.

 We then have
 \[|\beta_{7,2}(t) - \frac{1}{2}|^2 \leq \frac{1}{4} \iff f_8(\cos t) := \frac{5 \cos\frac{5\pi}{38}}{361 \sin\frac{\pi}{19}}\cos t - \frac{84}{361} + \frac{\cos^2\frac{5\pi}{38}}{361 \sin^2\frac{\pi}{19}} \leq 0.
 \]
 Since $f_8$ is non positive on $[-1, 1]$, it follows that $f_8(\cos t)$ is non positive for all $t$, which shows that $\beta_{7,2}$ lies below $\overset{\frown}{C_2C_3}$.

Finally, as previously noticed, since $\beta_{13,1}$ is a line segment, it suffices to show that its extremities lies below $\overset{\frown}{C_2C_3}$ by concavity. We compute: $|z_{13} - \frac{1}{2}|^2 \approx 0.244 < \frac{1}{4}$ and $|w_{13} - \frac{1}{2}| \approx 0.1 < \frac{1}{4}$. Thus $\overset{\frown}{C_2C_3}$ lies above $\beta_{13,1}$, and as $\Re(w_{13}) \approx 0.502 > \frac{1}{2} = \Re(C_3)$, we deduce that $\overset{\frown}{C_2C_3}$ is included in $M_{19}$.
\end{itemize}

This ends the proof of Lemma~\ref{lem:short-tori}.

\bibliography{universal-torus-triang}

\end{document}